\documentclass[a4paper,reqno,11pt]{amsart}

\usepackage{amsmath, amsfonts, amssymb, amsthm, amscd}
\usepackage{graphicx}
\usepackage{psfrag}
\usepackage{perpage}
\usepackage{url}
\usepackage{color}
\usepackage{subfig}

\usepackage[utf8]{inputenc}
\usepackage[T1]{fontenc}

\usepackage[a4paper,scale={0.72,0.74},marginratio={1:1},footskip=7mm,headsep=10mm]{geometry}

\usepackage{hyperref}

\numberwithin{equation}{section}

\newtheorem{theorem}{Theorem}[section]
\newtheorem{lemma}[theorem]{Lemma}
\newtheorem{proposition}[theorem]{Proposition}

\newtheorem{remark}[theorem]{Remark}

\newtheorem{hypothesis}[theorem]{Hypothesis}

\newcommand{\cA}{{\ensuremath{\mathcal A}} }

\newcommand{\cO}{{\ensuremath{\mathcal O}} }

\renewcommand{\tilde}{\widetilde}          
\DeclareMathSymbol{\leqslant}{\mathalpha}{AMSa}{"36} 
\DeclareMathSymbol{\geqslant}{\mathalpha}{AMSa}{"3E} 
\DeclareMathSymbol{\eset}{\mathalpha}{AMSb}{"3F}     
\newcommand{\dd}{\text{\rm d}}             

\newcommand{\R}{\mathbb{R}}
\newcommand{\C}{\mathbb{C}}

\newcommand{\N}{\mathbb{N}}

\newcommand{\PEfont}{\mathrm}

\renewcommand{\P}{\ensuremath{\PEfont P}}

\newcommand{\E}{\ensuremath{\PEfont E}}

\newcommand{\sign}{\mathrm{sign}}

\newcommand{\ind}{{\sf 1}}

\renewcommand{\epsilon}{\varepsilon} 
\renewcommand{\theta}{\vartheta} 
\renewcommand{\rho}{\varrho} 

\newenvironment{myenumerate}{%
\renewcommand{\theenumi}{\arabic{enumi}}%
\renewcommand{\labelenumi}{{\rm(\theenumi)}}%
\begin{list}{\labelenumi}
	{%
	\setlength{\itemsep}{0.4em}%
	\setlength{\topsep}{0.5em}%
	\setlength\leftmargin{2.45em}%
	\setlength\labelwidth{2.05em}%
	\setlength{\labelsep}{0.4em}%
	\usecounter{enumi}%
	}%
	}%
{\end{list}
}

\newenvironment{ienumerate}{%
\renewcommand{\theenumi}{\roman{enumi}}%
\renewcommand{\labelenumi}{{\rm(\theenumi)}}%
\begin{list}{\labelenumi}
	{%
	\setlength{\itemsep}{0.4em}%
	\setlength{\topsep}{0.5em}%
	\setlength\leftmargin{2.45em}%
	\setlength\labelwidth{2.05em}%
	\setlength{\labelsep}{0.4em}%
	\usecounter{enumi}%
	}%
	}%
{\end{list}
}

\newenvironment{aenumerate}{%
\renewcommand{\theenumi}{\alph{enumi}}%
\renewcommand{\labelenumi}{{\rm(\theenumi)}}%
\begin{list}{\labelenumi}
	{%
	\setlength{\itemsep}{0.4em}%
	\setlength{\topsep}{0.5em}%
	\setlength\leftmargin{2.45em}%
	\setlength\labelwidth{2.05em}%
	\setlength{\labelsep}{0.4em}%
	\usecounter{enumi}%
	}%
	}%
{\end{list}
}

\renewenvironment{enumerate}{
\begin{myenumerate}}%
{\end{myenumerate}}

\newenvironment{myitemize}{%
\begin{list}{$\bullet$}%
 	{%
	\setlength{\itemsep}{0.4em}%
	\setlength{\topsep}{0.5em}%
	\setlength\leftmargin{2.45em}%
	\setlength\labelwidth{2.05em}%
	\setlength{\labelsep}{0.4em}%
	}%
	}%
{\end{list}}

\renewenvironment{itemize}{
\begin{myitemize}}%
{\end{myitemize}}

\MakePerPage[2]{footnote} 

\title{}

\date{\today}

\newcommand\CBS{\mathsf{C_{BS}}}
\newcommand\VBS{\mathsf{V_{BS}}}
\newcommand\imp{\mathrm{imp}}
\newcommand\F{\overline F}

\newcommand\bkappa{\boldsymbol{\kappa}}

\title{General smile asymptotics with bounded maturity}

\author{Francesco Caravenna}
\address{Dipartimento di Matematica e Applicazioni, 
Universit\`a degli Studi di Milano-Bicocca,
via Cozzi 55,
I-20125 Milano, Italy}

\email{francesco.caravenna@unimib.it}

\author{Jacopo Corbetta}
\address{ 
\'{E}cole des Ponts - ParisTech,
CERMICS,
6 et 8 avenue Blaise Pascal,
77420 Champs sur Marne, 
France}
\email{jacopo.corbetta@enpc.fr}


\begin{document}


\begin{abstract}
We provide explicit conditions on the distribution of risk-neutral
log-returns which yield sharp asymptotic estimates on the
implied volatility smile. We allow for a variety of asymptotic regimes,
including both small maturity (with arbitrary strike) and
extreme strike (with arbitrary bounded maturity),
extending previous work of Benaim and Friz \cite{cf:BF09}.
We present applications to popular models, including Carr-Wu 
finite moment logstable model, Merton's jump diffusion model
and Heston's model.
\end{abstract}

\maketitle


\section{Introduction}

The price of a European option is typically expressed in terms of the
Black\&Scholes \emph{implied volatility} $\sigma_\imp(\kappa,t)$
(where $\kappa$ denotes the log-strike and $t$ the maturity), cf.\ \cite{cf:Gat06}.
Since exact formulas for a given model are typically out of reach, 
an active line or research is devoted to finding
asymptotic expansions for $\sigma_\imp(\kappa,t)$,
which can be useful in many respects,
e.g.\ for a fast calibration of some parameters of the model.
Explicit asymptotic formulas for $\sigma_\imp(\kappa,t)$
also allow to understand how the parameters affect
key features of the volatility surface, such as its slope, and what are the
possible shapes that can actually be obtained for a given model.
Let us mention the celebrated
Lee moment's formula~\cite{cf:L} and
more recent results 
\cite{cf:BF08, cf:BF09, cf:T09, cf:G, cf:FF,
cf:MT12, cf:GL,cf:FJ09,cf:RR, cf:GMZ14}.

\smallskip

A key problem is to link the implied volatility \emph{explicitly} to the
distribution of the risk-neutral log-return $X_t$, because the latter can be 
computed or estimated for many models.
The results of Benaim and Friz \cite{cf:BF09} are particularly
appealing, because they 
connect directly the asymptotic behavior of $\sigma_\imp(\kappa,t)$ to the
\emph{tail probabilities} 
\begin{equation}\label{ch2:eq:tail}
	\F_t(\kappa) := \P(X_t > \kappa) \,, \qquad
	F_t(-\kappa) := \P(X_t \le -\kappa) \,.
\end{equation}
Their results, which are limited to the special regime of extreme strike 
$\kappa \to \pm\infty$ with fixed maturity $t > 0$,
are based on the key notion
of \emph{regular variation}, which is appropriate when one considers
a single random variable $X_t$ (since $t$ is fixed).
This leaves out many interesting regimes,
notably the much studied case of small maturity $t \to 0$ with fixed strike $\kappa$.

\smallskip

In this paper we provide a substantial extension of \cite{cf:BF09}:
we formulate a suitable generalization of the regular variation assumption 
on $\F_t(\kappa)$, $F_t(\kappa)$ which, coupled to suitable moment conditions,
yields the 
asymptotic behavior of $\sigma_\imp(\kappa,t)$ in essentially \emph{any regime
of small maturity and/or extreme strike} (with bounded maturity).
We thus provide a unified approach, which includes as special cases both the regime of 
extreme strike $\kappa \to \pm\infty$ with fixed maturity $t > 0$, and
that of small maturity $t \to 0$ with fixed strike $\kappa$.
Mixed regimes, where $\kappa$ and $t$ vary simultaneously, are also allowed.
This flexibility yields asymptotic formulas for the
volatility surface $\sigma_\imp(\kappa,t)$ in open regions of the plane.

In Section~\ref{ch2:sec:examples} we illustrate our results 
through applications to 
popular models, such as Carr-Wu finite moment logstable model
and Merton's jump diffusion model.
We also discuss Heston's model, cf.\ \S\ref{ch2:sec:Heston}.
In a separate paper~\cite{cf:CC} we consider a stochastic volatility
model which exhibits multiscaling of moments, introduced in \cite{cf:ACDP}.

\smallskip

The key point in our analysis 
is to connect explicitly
the asymptotic behavior of the tail probabilities $\F_t(\kappa), F_t(\kappa)$
to call and put option prices $c(\kappa,t), p(\kappa,t)$
(cf.\ Theorems~\ref{ch2:th:main2b}, \ref{ch2:th:main2bl} and \ref{ch2:th:main2a}).
In fact, once the asymptotics of $c(\kappa,t), p(\kappa,t)$ are known,
the behavior of the implied volatility $\sigma_\imp(\kappa,t)$
can be deduced in a model independent way, as recently shown Gao and Lee~\cite{cf:GL}.
We summarize their results in~\S\ref{ch2:sec:main1}
(see Theorem~\ref{ch2:th:main1}),
where we also give an extension to a special regime, that was left out from their analysis
(cf.\ also \cite{cf:MT12}).

\smallskip

The paper is structured as follows.
\begin{itemize}
\item In Section~\ref{ch2:sec:main} we set some notation and
we state our main results.

\item In Section~\ref{ch2:sec:examples} we apply our results to some
popular models.

\item In Section~\ref{ch2:sec:pricetovol} we 
prove Theorem~\ref{ch2:th:main1}, linking option price and implied volatility.

\item In Section~\ref{ch2:sec:probtoprice} we prove our main results
(Theorems~\ref{ch2:th:main2b}, \ref{ch2:th:main2bl} and~\ref{ch2:th:main2a}). 

\item Finally, a few technical points
have been deferred to the Appendix~\ref{ch2:sec:app}.
\end{itemize}

\section{Main results}
\label{ch2:sec:main}

\subsection{The setting}
\label{ch2:sec:setting}

We consider a generic stochastic process $(X_t)_{t\ge 0}$
representing the log-price of an asset,
normalized by $X_0 := 0$.
We work under the risk-neutral measure,
that is (assuming zero interest rate) the price
process $(S_t := e^{X_t})_{t\ge 0}$ is a martingale.
European call and put options,
with maturity $t > 0$ and a log-strike $\kappa \in \R$, are
priced respectively
\begin{equation} \label{ch2:eq:cp}
	c(\kappa,t) = \E[(e^{X_t}-e^\kappa)^+] \,, \qquad
	p(\kappa,t) = \E[(e^\kappa - e^{X_t})^+] \,,
\end{equation}
and are linked by the \emph{call-put parity} relation:
\begin{equation}\label{ch2:eq:parity}
	c(\kappa,t) - p(\kappa,t) = 1-e^\kappa \,.
\end{equation}

\smallskip

As in \cite{cf:GL},
in our results
\emph{we take limits along an arbitrary family
(or ``path'') of values of $(\kappa,t)$}.
This includes both sequences $((\kappa_n,t_n))_{n\in\N}$ 
and curves $((\kappa_s, t_s))_{s \in [0,\infty)}$, 
hence we omit subscripts.
Without loss of generality, we assume that all the $\kappa$'s have the same sign
(just consider separately 
the subfamilies with positive and negative $\kappa$'s). To simplify notation,
we only consider positive families $\kappa \ge 0$
and give results for both $\kappa$ and $-\kappa$.

Our main interest is for families of values of $(\kappa,t)$ such that
\begin{equation}\label{ch2:eq:assfamily}
	\text{either $\kappa \to \infty$ with bounded $t$\,, \quad
	or $t \to 0$ with arbitrary $\kappa \ge 0$} \,.
\end{equation}
Whenever this holds, one has (see \S\ref{ch2:sec:cpexplained})
\begin{equation}\label{ch2:eq:cpvanish}
	c(\kappa,t) \to 0\,, \qquad 	p(-\kappa,t) \to 0 \,.
\end{equation}
We stress that \eqref{ch2:eq:assfamily} gathers many interesting regimes, namely:
\begin{aenumerate}
\item\label{ch2:it:a} $\kappa \to \infty$ and $t \to \bar t \in (0,\infty)$
(in particular, the case of fixed $t = \bar t > 0$);

\item\label{ch2:it:b} $\kappa \to \infty$ and $t \to 0$;

\item\label{ch2:it:c} $t \to 0$ and $\kappa \to \bar\kappa \in (0,\infty)$
(in particular, the case of fixed $\kappa = \bar \kappa > 0$);

\item\label{ch2:it:d} $t \to 0$ and $\kappa \to 0$.
\end{aenumerate}
Remarkably, while regime \eqref{ch2:it:d}
needs to be handled separately,
regimes \eqref{ch2:it:a}-\eqref{ch2:it:b}-\eqref{ch2:it:c} will be analyzed at once,
as special instances of the case ``$\kappa$ is bounded away from zero''.

\begin{remark}\rm
We stress the requirement of \emph{bounded} maturity $t$ in \eqref{ch2:eq:assfamily}.
Some of our arguments can be adapted to deal with cases when
$t \to \infty$, but additional work is needed
(for instance, we assume the boundedness of some exponential moments 
$\E[e^{(1+\eta)X_t}]$,
cf.\ \eqref{ch2:eq:moment}-\eqref{ch2:eq:momentsimple} below, which
is satisfied by most models if $t$ is bounded, but not if $t \to \infty$).
We refer to \cite{cf:T09,cf:JKM13} for results in the regime $t\to\infty$.
\end{remark}

Given a model $(X_t)_{t \in [0,\infty)}$, the \emph{implied volatility $\sigma_\imp(\kappa,t)$} is
defined as the value of the volatility parameter
$\sigma \in [0,\infty)$ that plugged into the Black\&Scholes formula
yields the given call and put prices $c(\kappa,t)$ and $p(\kappa,t)$ in \eqref{ch2:eq:cp}
(see \S\ref{ch2:sec:BS}-\S\ref{ch2:sec:th:mainmain} below).
To avoid trivialities, we focus on
families of $(\kappa,t)$ such that $c(\kappa,t) > 0$ and $p(-\kappa,t) > 0$
(in fact, note that $\sigma_\imp(\kappa,t) = 0$ if $c(\kappa,t) = 0$
and, likewise, $\sigma_\imp(-\kappa,t) = 0$ if $p(-\kappa,t) = 0$).

\medskip

\noindent
\emph{Notation.}
Throughout the paper, we write
$f(\kappa, t) \sim g(\kappa, t)$ to mean $f(\kappa, t) / g(\kappa, t) \to 1$.
Let us recall a useful standard device (\emph{subsequence argument}): 
to prove an asymptotic relation, such as e.g.\
$f(\kappa,t) \sim g(\kappa, t)$,
it suffices to show that from every subsequence 
one can extract a further sub-subsequence along which
the given relation holds. As a consequence, 
\emph{in the proofs we may always assume 
that all quantities of interest have a (possibly infinite) limit}, e.g.\ $\kappa \to
\bar\kappa \in [0,\infty]$ and $t \to \bar t \in [0,\infty)$,
because this is true along a suitable subsequence.

\subsection{Main results: atypical deviations}
\label{ch2:sec:main2}

We first focus on families
of $(\kappa,t)$ such that
\begin{equation} \label{ch2:eq:Fto0}
	\F_t(\kappa) \to 0 \,,
	\quad \ \ \text{resp.} \quad \ \
	F_t(-\kappa) \to 0 \,,
\end{equation}
a regime that we call \emph{atypical deviations}.
This is the most interesting case, much studied in the literature,
since it includes regimes
\eqref{ch2:it:a}, \eqref{ch2:it:b} and \eqref{ch2:it:c}
described on page \pageref{ch2:it:a},
and also regime \eqref{ch2:it:d} provided
$\kappa\to 0$ sufficiently slow.

When $\kappa \to \infty$ with fixed $t > 0$,
Benaim and Friz~\cite{cf:BF09}
require the \emph{regular variation} of the tail probabilities, i.e.\
there exist $\alpha > 0$ and a slowly varying 
function\footnote{A positive function $L(\cdot)$ is slowly varying
if $\lim_{x\to\infty} L(\rho x)/L(x) = 1$ for all $\rho > 0$.}
$L_t(\cdot)$ such that
\begin{equation} \label{ch2:eq:regvar}
\begin{split}
	& \log \F_t(\kappa) \sim -L_t(\kappa) \, \kappa^\alpha \,,
	\quad \ \ \text{resp.} \quad \ \
	\log F_t(-\kappa) \sim -L_t(\kappa) \, \kappa^\alpha \,.
\end{split}
\end{equation}
It is not obvious how to generalize \eqref{ch2:eq:regvar} when $t$ is allowed to vary,
i.e.\ which conditions to impose on $L_t(\kappa)$.
However, one can reformulate the first relation in \eqref{ch2:eq:regvar} 
simply requiring the
existence of $\lim_{\kappa\to\infty}\log \F_t(\rho\kappa) / \log\F_t(\kappa)$
for any fixed $\rho > 0$, by\ \cite[Theorem~1.4.1]{cf:BinGolTeu},
and analogously for the second relation in \eqref{ch2:eq:regvar}.
This reformulation (in which $L_t(\kappa)$ is not even mentioned!)
turns out to be the right condition to impose in the general context
that we consider, when $t$ is allowed to vary.
We are thus led to the following:

\begin{hypothesis}[Regular decay of tail probability]\label{ch2:ass:rv}
The family of values of $(\kappa,t)$
with $\kappa > 0$, $t > 0$
satisfies \eqref{ch2:eq:Fto0}, and for every $\rho \in [1,\infty)$ 
the following limit exists in $[0,+\infty]$:
\begin{equation} \label{ch2:eq:rv}
\begin{split}
	I_+(\rho) := \lim \frac{\log \F_t(\rho\kappa)}{\log \F_t(\kappa)} \,,
	\qquad \text{resp.} \qquad
	I_-(\rho) := \lim \frac{\log F_t(-\rho\kappa)}{\log F_t(-\kappa)} \,,
\end{split}
\end{equation}
where limits are taken along
the given family of values of $(\kappa,t)$.
Moreover 
\begin{equation}\label{ch2:eq:cont1}
	\lim_{\rho\downarrow 1} I_+(\rho) = 1 \,, \qquad
	\text{resp.} \qquad
	\lim_{\rho\downarrow 1} I_-(\rho) = 1 \,.
\end{equation}
\end{hypothesis}

Depending on the regime of $\kappa$, we will also need
one of the following \emph{moment conditions}.
\begin{itemize}
\item Given $\eta \in (0,\infty)$, the first moment condition is
\begin{gather} \label{ch2:eq:moment}
	\limsup\, \E[e^{(1+\eta)X_t}] < \infty \,,
\end{gather}
along the given family of values of $(\kappa,t)$.
When $t \le T$, it is enough to require that
\begin{gather} \label{ch2:eq:momentsimple}
	\E[e^{(1+\eta)X_T}] < \infty \,,
\end{gather}
because $(e^{(1+\eta)X_t})_{t\ge 0}$ is a submartingale
and hence $\E[e^{(1+\eta)X_t}] \le \E[e^{(1+\eta)X_T}]$.

\item Given $\eta \in (0,\infty)$, the second moment condition is
\begin{gather}
	\label{ch2:eq:moment0p}
	\limsup\,
	\E\bigg[ \bigg|\frac{e^{X_t} - 1}{\kappa} \bigg|^{1+\eta} 
	\bigg] < \infty \,,
\end{gather}
along the given family
of values of $(\kappa,t)$.
Note that for $\eta = 1$ this simplifies to
\begin{equation}\label{ch2:eq:moment0psimple}
	\exists C \in (0,\infty): \qquad
\E [e^{2 X_t} ] \le 1 + C \kappa^2 \,.
\end{equation}

\end{itemize}

We are ready to state our main results, which
express the asymptotic behavior of option prices
and implied volatility explicitly in terms of the tail probabilities.
Due to different assumptions, we first consider
right-tail asymptotics.

\begin{theorem}[Right-tail atypical deviations]\label{ch2:th:main2b}
Consider a family of values of $(\kappa,t)$ with $\kappa > 0$, $t > 0$ such that
Hypothesis~\ref{ch2:ass:rv} is satisfied by the right tail probability 
$\F_t(\kappa)$.

\begin{ienumerate}
\item {\sf [$\kappa$ bounded away from zero, $t$ bounded away from infinity
($\liminf\kappa > 0$, $\limsup t < \infty$)]}
Let the moment condition \eqref{ch2:eq:moment} hold for \emph{every $\eta > 0$},
or alternatively let it hold
only for \emph{some $\eta > 0$} but in addition assume that
\begin{equation}\label{ch2:eq:Iplus}
	I_+(\rho) \ge \rho \,, \qquad \forall \rho \ge 1 \,.
\end{equation}
Then
\begin{align}\label{ch2:eq:ma1c}
	\log c(\kappa,t) & \sim  \log \F_t(\kappa) + \kappa \,, \\
	\label{ch2:eq:sigmainf}
	\sigma_\imp(\kappa,t) & \sim
	\left( \sqrt{\frac{-\log\F_t(\kappa)}{\kappa}}
	- \sqrt{\frac{-\log\F_t(\kappa)}{\kappa}-1}\, \right)
	\sqrt{\frac{2\kappa}{t}} \,.
\end{align}
{\sf Special case:} if $-\log \F_t(\kappa)/\kappa \to \infty$,
assumption \eqref{ch2:eq:Iplus} can be relaxed to
\begin{equation}\label{ch2:eq:I+infty}
	\lim_{\rho\to\infty} I_+(\rho) = \infty \,, 
\end{equation}
and relations \eqref{ch2:eq:ma1c}-\eqref{ch2:eq:sigmainf} simplify to
\begin{align}\label{ch2:eq:mac}
	\log c(\kappa,t) & \sim  \log \F_t(\kappa) \,, \\
	\label{ch2:eq:sigmafin}
	\sigma_\imp(\kappa,t) & \sim 
	\frac{\kappa}{\sqrt{2t\, (-\log \F_t(\kappa) )}} \,.
\end{align}

\item {\sf [$\kappa$ and $t$ vanish ($\kappa \to 0$, $t \to 0$)]} 
Let the moment condition \eqref{ch2:eq:moment0p} hold for \emph{every $\eta > 0$},
or alternatively let it hold only for \emph{some $\eta > 0$} but in addition assume
\eqref{ch2:eq:I+infty}. Then
\begin{align} \label{ch2:eq:ma2c}
	\log \big( c(\kappa,t) / \kappa \big) & \sim  \log \F_t(\kappa) \,,\\
	\label{ch2:eq:sigmafinbis}
	\sigma_\imp(\kappa,t) & \sim 
	\frac{\kappa}{\sqrt{2t\, (-\log \F_t(\kappa) )}} \,.
\end{align}
\end{ienumerate}
\end{theorem}

Next we turn to left-tail asymptotics.
The assumptions in this case turn out to be sensibly weaker than those for right-tail.
For instance, the left-tail condition
$\E[e^{-\eta X_T}] < \infty$ required in \cite[Theorem 1.2]{cf:BF09}
is not needed, which allows to treat the case of a \emph{polynomially decaying} left tail, 
like in the Carr-Wu model described
in Section~\ref{ch2:sec:examples}.

\begin{theorem}[Left-tail atypical deviations]\label{ch2:th:main2bl}
Consider a family of values of $(\kappa,t)$ with $\kappa > 0$, $t > 0$ such that
Hypothesis~\ref{ch2:ass:rv} is satisfied by the left tail probability $F_t(-\kappa)$.

\begin{itemize}
\item {\sf [$\kappa$ bounded away from zero, $t$ bounded away from infinity
($\liminf\kappa > 0$, $\limsup t < \infty$)]}
With no moment condition and no extra assumption on $I_-(\cdot)$, one has
\begin{align}\label{ch2:eq:ma1p}
	\log p(-\kappa,t) & \sim  \log F_t(-\kappa) - \kappa \,, \\
	\label{ch2:eq:sigmainfl}
	\sigma_\imp(-\kappa,t) & \sim
	\left( \sqrt{\frac{-\log F_t(-\kappa)}{\kappa} + 1}
	- \sqrt{\frac{-\log F_t(-\kappa)}{\kappa}}\, \right)
	\sqrt{\frac{2\kappa}{t}} \,.
\end{align}
{\sf Special case:} if $-\log F_t(-\kappa)/\kappa \to \infty$,
relations \eqref{ch2:eq:ma1p}-\eqref{ch2:eq:sigmainfl} simplify to
\begin{align}\label{ch2:eq:map}
	\log p(-\kappa,t) & \sim  \log F_t(-\kappa) \,,\\
	\label{ch2:eq:sigmafinl}
	\sigma_\imp(-\kappa,t) & \sim 
	\frac{\kappa}{\sqrt{2t\, (-\log F_t(-\kappa) )}} \,.
\end{align}

\item {\sf [$\kappa$ and $t$ vanish ($\kappa \to 0$, $t \to 0$)]} 
Let the moment condition \eqref{ch2:eq:moment0p} hold for \emph{every $\eta > 0$},
or alternatively let it hold only for \emph{some $\eta > 0$} but in addition assume that
\begin{equation}\label{ch2:eq:I-infty}
	\lim_{\rho\uparrow\infty} I_-(\rho) = \infty \,.
\end{equation}
Then
\begin{align}\label{ch2:eq:ma2p}
	\log \big( p(-\kappa,t) / \kappa \big) & \sim  \log F_t(-\kappa) \,, \\
	\label{ch2:eq:sigmafinlbis}
	\sigma_\imp(-\kappa,t) & \sim 
	\frac{\kappa}{\sqrt{2t\, (-\log F_t(-\kappa) )}} \,.
\end{align}
\end{itemize}
\end{theorem}

We prove Theorems~\ref{ch2:th:main2b} and~\ref{ch2:th:main2bl}
in \S\ref{ch2:sec:main2bproof} below. The key step is to link the
option prices $c(\kappa,t)$, $p(\kappa,t)$ to the tail probabilities
$\F_t(\kappa)$, $F_t(-\kappa)$, exploiting Hypothesis~\ref{ch2:ass:rv}.
Once this is done, the asymptotic behavior of the implied volatility
$\sigma_\imp(\kappa,t)$ can be deduced using the model independent results
of \cite{cf:GL}, that we summarize in \S\ref{ch2:sec:main1}.

\begin{remark}\rm\label{ch2:rem:moment}
The ``special case'' conditions
\begin{equation}\label{ch2:eq:condit}
	-\frac{\log \F_t(\kappa)}{\kappa} \to \infty \,, \qquad
	\text{resp.} \qquad
	-\frac{\log F_t(-\kappa)}{\kappa} \to \infty \,,
\end{equation}
are automatically fulfilled in the small maturity regime $t \to 0$
with fixed strike $\kappa = \bar\kappa > 0$.
In this case, one can use the simplified formulas \eqref{ch2:eq:mac}-\eqref{ch2:eq:sigmafin}
and \eqref{ch2:eq:map}-\eqref{ch2:eq:sigmafinl}.
\end{remark}

\subsection{Main results: typical deviations.}

Next we focus on the regime when $t\to 0$ and $\kappa \to 0$ sufficiently fast,
so that the tail probability $\F_t(\kappa)$, resp. $F_t(-\kappa)$
has a strictly positive limit and condition \eqref{ch2:eq:Fto0}
is violated. We call this regime \emph{typical deviations}.
This includes the basic regime of fixed $\kappa = 0$
and $t \downarrow 0$. 
Mixed regimes, when $\kappa \to 0$ and $t \to 0$
simultaneously, are also interesting,
e.g.\ to interpolate between the at-the-money $(\kappa = 0)$ and
out-of-the-money $(\kappa \ne 0)$ cases, which can be strikingly different as $t \to 0$
(see \cite{cf:MT12}).

We make the following natural assumption.

\begin{hypothesis}[Small time scaling] \label{ch2:ass:smalltime}
There is a positive function $(\gamma_t)_{t > 0}$ with
$\lim_{t\downarrow 0}\gamma_t = 0$ such
that $X_t / \gamma_t$ converges in law as $t \downarrow 0$ to
some random variable $Y$:
\begin{equation} \label{ch2:eq:scaling}
	\frac{X_t}{\gamma_t} \xrightarrow[t\downarrow 0]{d} Y \,.
\end{equation}
\end{hypothesis}

We refer to Remark~\ref{rem:equivalent} below for concrete ways to check
Hypothesis~\ref{ch2:ass:smalltime}.
Let us stress that \eqref{ch2:eq:scaling} is a condition on the
tail probabilities, since it can be reformulated as
\begin{equation} \label{ch2:eq:weakconv}
	\F_t(a\gamma_t) \to \P(Y > a) \qquad \text{and} \qquad
	F_t(-a\gamma_t) \to \P(Y \le -a) \,,
\end{equation}
for all $a\ge 0$ with $\P(|Y|=a) = 0$.
If the support of the law of $Y$ is unbounded from above and below
(as it is usually the case), the limits in \eqref{ch2:eq:weakconv} are strictly positive
for every $a \ge 0$.

\smallskip

The appropriate moment condition in this regime turns out to be
\eqref{ch2:eq:moment0p} with $\kappa = \gamma_t$, i.e.
\begin{equation}\label{ch2:eq:assX}
       \exists \eta > 0: \qquad
	\limsup_{t\to 0} \, \E\bigg[ \bigg|\frac{e^{X_t} - 1}{\gamma_t} \bigg|^{1+\eta} 
	\bigg] < \infty \,.
\end{equation}

Lastly, we introduce some notation.
Denote by $\phi(\cdot)$
and $\Phi(\cdot)$ 
respectively the density and distribution function of a standard Gaussian
(see \eqref{ch2:eq:phiPhi} below), and define the function
\begin{equation} \label{ch2:eq:D}
	D(z) := \frac{\phi(z)}{z} - \Phi(-z) , \qquad \forall z > 0 \,.
\end{equation}
As we show in \S\ref{ch2:eq:backg} below,
$D$ is a smooth and strictly decreasing bijection from $(0,\infty)$ to $(0,\infty)$.
Its inverse $D^{-1}:(0,\infty) \to (0,\infty)$ is also smooth, strictly decreasing and satisfies
\begin{equation}\label{ch2:eq:Das}
	D^{-1}(y) \sim \sqrt{2 \, (-\log y)} \quad \
	\text{as } y \downarrow 0  \,, \qquad \ \
	D^{-1}(y) \sim \frac{1}{\sqrt{2\pi}} \frac{1}{y} \quad \
	\text{as } y \uparrow \infty\,.
\end{equation}

We can finally state the following result,
linking option prices
and implied volatility to tail probabilities in the regime of typical deviations.

\begin{theorem}[Typical deviations]\label{ch2:th:main2a}
Assume that Hypothesis~\ref{ch2:ass:smalltime} is satisfied,
and moreover the moment condition \eqref{ch2:eq:assX} holds.
Then the random variable in \eqref{ch2:eq:scaling} satisfies
$\E[Y]=0$.

Fix $a \in [0,\infty)$ with $\P(Y > a) > 0$, resp.\ $\P(Y < -a) > 0$.
For any family of $(\kappa,t)$ with
\begin{equation*}
	t \to 0 \ \text{ and } \ \frac{\kappa}{\gamma_t} \to a \in [0,\infty)  \,,
\end{equation*}
the asymptotic behavior of option prices is given by
\begin{align} \label{ch2:eq:astypical}
	c(\kappa, t) & \sim \gamma_t \, \E[(Y-a)^+]  \,,
	\qquad \text{resp.} \qquad
	p(-\kappa, t) \sim \gamma_t \, \E[(Y+a)^-] \,,
\end{align}
and correspondingly the implied volatility is given by
\begin{align} \label{ch2:eq:sigma0}
	\sigma_\imp(\pm\kappa,t) & \sim	
	C_\pm(a) \, \frac{\gamma_t}{\sqrt{t}} \,, \ \quad
	\text{with} \quad \ C_\pm(a) =
	\begin{cases}
	\displaystyle\frac{a}{\rule{0pt}{1.1em}D^{-1}\big(\frac{\E[(Y\mp a)^\pm]}{a} \big)}
	& \text{ if \ $\displaystyle
	a > 0 $} \,, \\
	\rule{0em}{1.7em}\displaystyle \sqrt{2\pi} \, \E[Y^\pm] 
	& \text{ if \ $\displaystyle
	a =  0$} \,.
	\end{cases}
\end{align}
\end{theorem}

\begin{remark}\rm\label{rem:equivalent}
Hypothesis~\ref{ch2:ass:smalltime} can be easily checked when the characteristic
function of $X_t$ is known,
because, by L\'evy continuity theorem, 
the convergence in distribution \eqref{ch2:eq:scaling} is equivalent to
the pointwise convergence
$\E[e^{i u X_t/\gamma_t}] \to \E[e^{i u Y}]$ for every $u \in \R$.
We will see concrete examples in
Subsections~\ref{ch2:sec:CW} (Carr-Wu model)
and~\ref{ch2:sec:Merton} (Merton's model).

\smallskip

Another interesting case is that of \emph{diffusions}. Assume that $X_t = \log S_t$, where
$(S_t)_{t\ge 0}$ evolves according to the stochastic differential equation 
\begin{equation}\label{eq:sde}
	\begin{cases}
	\dd S_t = \sqrt{V_t} \, S_t \, \dd W_t \\
	S_0 = 1
	\end{cases} \,,
\end{equation}
where $W = (W_t)_{t\ge 0}$ is a Brownian motion and 
$V = (V_t)_{t\ge 0}$ is a positive adapted
process, representing the volatility, 
possibly correlated with $W$.
Under the mild assumption that $\lim_{t\to 0} V_t = \sigma_0^2$ a.s.,
where $\sigma_0 \in (0,\infty)$ is a \emph{constant},
one can show that Hypothesis~\ref{ch2:ass:smalltime} holds with $\gamma_t = \sqrt{t}$
and $Y \sim N(0,\sigma_0^2)$, see Appendix~\ref{app:diffusions}.

Interestingly, plugging $Y \sim N(0,\sigma_0^2)$ into \eqref{ch2:eq:sigma0} yields
$C_\pm(a) \equiv \sigma_0$ (see Appendix~\ref{app:diffusions}).
Consequently, if the moment condition \eqref{ch2:eq:assX} holds,
we can apply Theorem~\ref{ch2:th:main2a},
getting $\sigma_\imp(\pm\kappa, t) \sim \sigma_0$ along any parabolic curve
$\kappa \sim a \sqrt{t}$. This result is consistent
with recent results by Pagliarani
and Pascucci~\cite{cf:PP15}, who go beyond first-order asymptotics.
\end{remark}

The proof of Theorem~\ref{ch2:th:main2a} is given in \S\ref{ch2:sec:proofth:main2a} below.
The asymptotic behavior \eqref{ch2:eq:astypical} of option prices follows easily from the convergence
in distribution \eqref{ch2:eq:scaling}, because the needed uniform integrability is ensured
by the moment condition \eqref{ch2:eq:assX}.
The asymptotic behavior \eqref{ch2:eq:sigma0} of implied volatility
can again by deduced from the option prices asymptotics in a model independent way,
that we now describe. 

\subsection{From option price to implied volatility}
\label{ch2:sec:main1}

Whenever the option prices $c(\kappa,t)$
or $p(-\kappa,t)$ vanish, they determine the asymptotic behavior
of the implied volatility through \emph{explicit universal formulas}.
These are summarized in the following theorem (of which we give in Section~\ref{ch2:sec:pricetovol}
a self-contained proof),
which gathers results from the recent literature.

\begin{theorem}[From option price to implied volatility]\label{ch2:th:main1}
Consider an arbitrary family of values of $(\kappa,t)$ with $\kappa\ge 0$ and $t>0$,
such that $c(\kappa, t) \to 0$, resp.\ $p(-\kappa, t) \to 0$.
\begin{itemize}
\item {\sf Case of $\kappa$ bounded away from zero
(i.e.\ $\liminf\kappa > 0$).}
\begin{equation} \label{ch2:eq:Vas>00}
\begin{split}
	\sigma_\imp(\kappa,t) & \sim
	\bigg( \sqrt{\frac{-\log c(\kappa,t)}{\kappa} + 1} 
	- \sqrt{\frac{-\log c(\kappa,t)}{\kappa}} \, \bigg)
	\sqrt{\frac{2\kappa}{t}} \,, \qquad \text{resp.} \\
	\sigma_\imp(-\kappa,t)
	& \sim \bigg( \sqrt{\frac{-\log p(-\kappa,t)}{\kappa}} 
	- \sqrt{\frac{-\log p(-\kappa,t)}{\kappa} - 1} \, \bigg)
	\sqrt{\frac{2\kappa}{t}}\,.
\end{split}
\end{equation}
\item {\sf Case of $\kappa\to 0$, with  $\kappa>0$.}
\begin{equation} \label{ch2:eq:Vas<infty0}
\begin{split}
	\sigma_\imp(\kappa,t) & \sim
	\frac{1}{\rule{0pt}{1.2em}D^{-1} \Big(\frac{c(\kappa,t)}{\kappa} \Big)}
	\, \frac{\kappa}{\sqrt{t}} \,, \qquad \text{resp.} \\
	\sigma_\imp(-\kappa,t) & \sim
	\frac{1}{\rule{0pt}{1.2em}D^{-1} \Big(\frac{p(-\kappa,t)}{\kappa} \Big)}
	\, \frac{\kappa}{\sqrt{t}} \,,
\end{split}
\end{equation}
where the function $D: (0,\infty) \to (0,\infty)$ is defined in \eqref{ch2:eq:D}-\eqref{ch2:eq:Das}.

\item {\sf Case of $\kappa = 0$.}
\begin{equation}\label{ch2:eq:Vas00}
	\sigma_\imp(0,t) \sim \sqrt{2\pi} \,\, \frac{c(0,t)}{\sqrt{t}}
	= \sqrt{2\pi} \,\, \frac{p(0,t)}{\sqrt{t}} \,.
\end{equation}
\end{itemize}
\end{theorem}

We stress that Theorem~\ref{ch2:th:main1} allows to derive
immediately all the asymptotic relations for the implied volatility
$\sigma_\imp(\pm\kappa,t)$ appearing in Theorems~\ref{ch2:th:main2b}, \ref{ch2:th:main2bl}
and~\ref{ch2:th:main2a} from the corresponding relations for the option
prices $c(\kappa,t)$ and $p(-\kappa,t)$.

The main part of Theorem~\ref{ch2:th:main1} is equation \eqref{ch2:eq:Vas>00},
which was recently proved
by Gao and Lee \cite{cf:GL} extending previous results of Lee \cite{cf:L},
Roper and Rutkowski~\cite{cf:RR},
Benaim and Friz~\cite{cf:BF09} and Gulisashvili~\cite{cf:G}.
As a matter of fact, Gao and Lee prove much more than \eqref{ch2:eq:Vas>00},
providing explicit estimates for the error beyond
first order asymptotics. 

Equation \eqref{ch2:eq:Vas<infty0} is a new contribution of the present paper.
In fact,  \cite{cf:GL} assume that $-\log \kappa = o( -\log c(\kappa,t) )$
(cf.\ equation (4.2) therein), which excludes the regimes with
$\kappa \to 0$ ``fast enough''. The relevance of such regimes
has been recently shown in \cite{cf:MT12}, where the special case 
$\kappa \propto \sqrt{t \log(1/t)}$ is considered
(see \cite[Theorem~3.1]{cf:MT12}).

\begin{remark}\rm\label{ch2:rem:simple}
Relation \eqref{ch2:eq:Vas<infty0} provides an \emph{interpolation
between the at-the-money and out-of-the-money regimes},
described by \eqref{ch2:eq:Vas00} 
and \eqref{ch2:eq:Vas>00}. Let us be more explicit.

Using \eqref{ch2:eq:Das},
formula \eqref{ch2:eq:Vas<infty0} can be rewritten as follows:
\begin{align}\label{ch2:eq:Vas<infty0simple}
	& \sigma_\imp(\kappa,t) \sim \begin{cases}
	\rule{0em}{1.6em}\displaystyle
	\frac{\kappa}{\sqrt{2t\, ( - \log ( c(\kappa,t) / \kappa ) )}} & 
	\text{ if \ 
		$\displaystyle\frac{c(\kappa,t)}{\kappa} \to 0$} \,; \\
	\rule[-1.4em]{0em}{3.5em}\displaystyle
	\frac{\kappa}{D^{-1}(a) \, \sqrt{t}}  & 
	\text{ if \ 
	$\displaystyle\frac{c(\kappa,t)}{\kappa} \to a \in (0,\infty)$} \,; \\
	\rule[-1.4em]{0em}{3.5em}\displaystyle
	\sqrt{2\pi} \, \frac{c(\kappa,t)}{\sqrt{t}}  & 
	\text{ if \ 
	$\displaystyle\frac{c(\kappa,t)}{\kappa} \to \infty$, \ or
	if \ $\kappa = 0$} \,,
	\end{cases}
\end{align}
and analogously for $\sigma_\imp(-\kappa,t)$, just replacing
$c(\kappa,t)$ by $p(-\kappa,t)$.

Note that the last line in \eqref{ch2:eq:Vas<infty0simple} matches
with the at-the-money regime \eqref{ch2:eq:Vas00}.
In order to see how \eqref{ch2:eq:Vas<infty0simple} matches
with the out-of-the-money regime \eqref{ch2:eq:Vas>00}, it suffices to
note that whenever $\frac{-\log c(\kappa,t)}{\kappa} \to \infty$, 
resp.\ $\frac{-\log p(-\kappa,t)}{\kappa} \to \infty$,
formula \eqref{ch2:eq:Vas>00} can be rewritten as
\begin{equation}\label{ch2:eq:Vas>00simple}
	\sigma_\imp(\kappa,t) \sim \frac{\kappa}{\sqrt{2t\, (-\log  c(\kappa,t) )}} \,,
	\quad \ \ \text{resp.} \quad \ \
	\sigma_\imp(-\kappa,t) \sim \frac{\kappa}{\sqrt{2t\, (-\log  p(-\kappa,t) )}} \,,
\end{equation}
and this coincides with the first line of 
\eqref{ch2:eq:Vas<infty0simple} when $\kappa \to 0$ slowly enough, namely
\begin{equation}\label{eq:BL}
	-\log \kappa = o\big( -\log c(\kappa,t) \big)\,.
\end{equation}
\end{remark}

\subsection{Discussion}

Theorems~\ref{ch2:th:main2b}, \ref{ch2:th:main2bl}
and \ref{ch2:th:main2a} are useful because their assumptions,
involving asymptotics for the tail probabilities
$\F_t(\kappa)$ and $F_t(-\kappa)$,
can be directly verified for many concrete models
(see Section~\ref{ch2:sec:examples} for some examples).
The difference between the regimes of typical and atypical deviations
can be described as follows:
\begin{itemize}
\item for typical deviations, the key assumption is
Hypothesis~\ref{ch2:ass:smalltime}, which
concerns the \emph{weak convergence}
of $X_t$, cf.\ \eqref{ch2:eq:scaling}-\eqref{ch2:eq:weakconv};

\item for atypical deviations,
the key assumption is Hypothesis~\ref{ch2:ass:rv}, which
concerns the \emph{large deviations} properties of $X_t$,
cf.\ \eqref{ch2:eq:rv}-\eqref{ch2:eq:cont1}.
\end{itemize}

In particular, it is worth stressing that 
Hypothesis~\ref{ch2:ass:rv}
requires sharp asymptotics only for \emph{the logarithm of the tail probabilities
$\log\F_t(\kappa)$ and $\log F_t(-\kappa)$}, and not for the tail
probabilities themselves, which would be a considerably harder task
(out of reach for many models). As a consequence,
Hypothesis~\ref{ch2:ass:rv} can often be checked
through the celebrated \emph{G\"artner-Ellis Theorem} \cite[Theorem 2.3.6]{cf:DZ},
which yields sharp asymptotics on $\log\F_t(\kappa)$ and $\log F_t(-\kappa)$
under suitable conditions on the moment generating function of $X_t$.

\section{Applications}
\label{ch2:sec:examples}

In this section we show the relevance of our main
theoretical results, deriving asymptotic expansions
of the implied volatility for Carr-Wu finite moment logstable model (\S\ref{ch2:sec:CW}) 
and Merton's jump diffusion model (\S\ref{ch2:sec:Merton}).
The case of Heston's model is briefly discussed in~\S\ref{ch2:sec:Heston}.

Our results can also be applied to a stochastic volatility model, 
recently introduced in \cite{cf:ACDP}, which exhibits multiscaling of moments.
Even though no closed expression is available for the moment generating 
function of the log-price, the tail probabilities can be estimated
\emph{explicitly}, as we show in a separate paper \cite{cf:CC}.
This leads to precise asymptotics for the implied volatility,
thanks to Theorems~\ref{ch2:th:main2b}, \ref{ch2:th:main2bl} and~\ref{ch2:th:main2a}.

\subsection{Carr-Wu Finite Moment Logstable Model}
\label{ch2:sec:CW}

Carr and Wu \cite{cf:CW04} consider a model
where the log-strike $X_t$ has characteristic function
\begin{equation} \label{ch2:eq:charX}
	  \E\big[e^{iuX_t}\big]=
	 e^{t\left[iu\mu - |u|^{\alpha}\sigma^{\alpha}
	\left(1 + i \, \sign (u) \tan ( \frac{\pi \alpha}{2}) \right) \right]} \,,
\end{equation}
where $\sigma \in (0,\infty)$, $\alpha \in (1,2]$
and we fix $\mu := \sigma^\alpha / \cos(\frac{\pi\alpha}{2})$
to work in the risk-neutral measure,
cf.\ \cite[Proposition 1]{cf:CW04}. 
The moment generating function
of $X_t$ is
\begin{equation} \label{ch2:eq:momgenCW}
	\E\big[ e^{\lambda X_t} \big] =
	\begin{cases}
	\displaystyle
	e^{[\lambda \mu -
	\frac{(\lambda\sigma)^{\alpha}}{\cos (\frac{\pi \alpha}{2} )} ] \, t} &
	\text{if } \lambda \ge 0 \,, \\
	+\infty & \text{if } \lambda < 0 \,.
	\end{cases} 
\end{equation}
Note that as $\alpha \to 2$ one recovers Black\&Scholes model 
with volatility $\sqrt{2}\sigma$, cf.\ \S\ref{ch2:sec:BS} below.

Applying Theorems~\ref{ch2:th:main2b}, \ref{ch2:th:main2bl} and \ref{ch2:th:main2a},
we give a \emph{complete characterization} of the volatility smile
asymptotics with bounded maturity. 
This
includes, in particular, the regimes
of extreme strike ($\kappa \to \pm\infty$ with fixed $t > 0$) and
of small maturity ($t \to 0$ with fixed $\kappa$).

\begin{theorem}[Smile asymptotics of Carr-Wu model]\label{th:CW}
The following asymptotics hold.
\begin{itemize}
\item {\sf Atypical deviations.} Consider any family of $(\kappa,t)$ with
$\kappa \ge 0$, $t > 0$ such 
that
\begin{equation} \label{ch2:eq:CWaty}
	\text{either} \qquad t \to 0 \ \ \text{and} \ \ \kappa \gg t^{1/\alpha} \,, \qquad
	\text{or} \qquad t \to \bar t \in (0,\infty) \ \
	\text{and} \ \ \kappa \to \infty \,.
\end{equation}
(This includes the regimes \eqref{ch2:it:a}, \eqref{ch2:it:b}, \eqref{ch2:it:c}
on page \pageref{ch2:it:a}, and part of regime \eqref{ch2:it:d}.)
Then one has the right-tail asymptotics
\begin{equation}\label{ch2:eq:CWinf}
	\sigma_\imp(\kappa,t) \sim 
	B_\alpha \bigg(\frac{\kappa}{t}\bigg)^{-\frac{2-\alpha}{2(\alpha-1)}} \,,
	\qquad \text{where} \ \
	B_\alpha :=
	\frac{(\alpha\sigma)^{\frac{\alpha/2}{\alpha-1}}}
	{\sqrt{2(\alpha-1)} \, |\cos(\frac{\pi \alpha}{2})|^{\frac{1/2}{\alpha-1}}}	\,.
\end{equation}
The corresponding left-tail asymptotics are given by
\begin{equation} \label{ch2:eq:qualef}
\begin{split}
	\sigma_\imp(-\kappa,t) & \sim
	\left( \sqrt{\frac{\log \frac{\kappa^\alpha}{t}}{\kappa} + 1}
	- \sqrt{\frac{\log \frac{\kappa^\alpha}{t}}{\kappa}}\, \right)
	\sqrt{\frac{2\kappa}{t}} \,,
\end{split}
\end{equation}
which can be made more explicit distinguishing different regimes:
\begin{equation}\label{ch2:eq:CWinfl}
	\sigma_\imp(-\kappa,t) \sim
	\begin{cases}
	\displaystyle
	\frac{\kappa}{\sqrt{2 t \log \frac{\kappa^\alpha}{t}}}
	& \ \ \, \text{if } \ t \to 0 \ \text{ and } \
	\displaystyle
	\frac{\kappa}{\log \frac{1}{t}} \to 0 \,,\\
	\displaystyle
	\rule{0pt}{2.2em}
	\left(\frac{\sqrt{1+a} - 1}{\sqrt{a}}\right)
	\sqrt{\frac{2\kappa}{t}}
	& \ \ \, \text{if } \ t \to 0 \ \text{ and } \
	\displaystyle
	\frac{\kappa}{\log \frac{1}{t}} \to a \in (0,\infty) \,,\\
	\displaystyle
	\rule{0pt}{3.0em}
	\sqrt{\frac{2\kappa}{t}}
	& 
	\displaystyle
	\begin{cases}
	\text{if } \ 
	t \to 0 \ \text{ and } \ 
	\displaystyle
	\frac{\kappa}{\log \frac{1}{t}} \to \infty \,, \\
	\rule{0pt}{1.3em}
	\text{if } \ 
	t \to \bar t \in (0,\infty) \ \text{ and } \ \kappa \to \infty \,.
	\end{cases} \,.
	\end{cases}
\end{equation}

\item {\sf Typical deviations.} For any
family of $(\kappa,t)$ with
\begin{equation} \label{ch2:eq:CWty}
	t \to 0\,, \qquad
	\frac{\kappa}{t^{1/\alpha}} \to a \in [0,\infty) \,,
\end{equation}
one has
\begin{equation} \label{ch2:eq:CWzero}
	\sigma_\imp(\pm\kappa, t) \sim C_\pm(a) \, t^{\frac{2-\alpha}{2\alpha}} \,,
	\quad \text{with} \quad
	C_\pm(a) := \begin{cases}
	\displaystyle
	\frac{a}{D^{-1}\big( \frac{\E[(\sigma Y \mp a)^\pm]}{a} \big)}
	& \text{if } a > 0 \,, \\
	\displaystyle
	\rule{0pt}{1.6em}
	\sqrt{2\pi} \, \sigma \,\E[Y^\pm]
	& \text{if } a = 0 \,.
	\end{cases}
\end{equation}
\end{itemize}
\end{theorem}

\begin{remark}[Surface asymptotics for the Carr-Wu model]\rm\label{ch2:rem:joint}
The fact that relations \eqref{ch2:eq:CWinf} and \eqref{ch2:eq:qualef}
hold for \emph{any} family of $(\kappa,t)$ satisfying \eqref{ch2:eq:CWaty}
yields interesting consequences. We claim that, for any $T \in (0,\infty)$ and
$\epsilon > 0$, there exists $M = M(\epsilon,T) \in (0, \infty)$ such that
the following inequalities hold
\emph{for all $(\kappa,t)$ in the region $\cA_{T,M} := \{0 < t \le T, \ \kappa > M t^{1/\alpha}\}$}:
\begin{equation} \label{ch2:eq:jointsurface}
	\big(1-\epsilon\big)
	B_\alpha \bigg(\frac{\kappa}{t}\bigg)^{-\frac{2-\alpha}{2(\alpha-1)}} \le
	\sigma_\imp(\kappa,t) \le \big(1+\epsilon\big)
	B_\alpha \bigg(\frac{\kappa}{t}\bigg)^{-\frac{2-\alpha}{2(\alpha-1)}} \,.
\end{equation}
Similar inequalities can be deduced from
\eqref{ch2:eq:qualef}-\eqref{ch2:eq:CWinfl} and  \eqref{ch2:eq:CWzero}.
Relation \eqref{ch2:eq:jointsurface} 
gives a \emph{uniform approximations} of the volatility surface $\sigma_\imp(\kappa,t)$
in open regions of the plane $(\kappa,t)$.

The proof of \eqref{ch2:eq:jointsurface} is simple:
assume by contradiction that there exist $T, \epsilon \in (0,\infty)$
such that for every $M \in (0,\infty)$ relation \eqref{ch2:eq:jointsurface}
\emph{fails} for some $(\kappa_M, t_M) \in \cA_{T,M}$. Extracting a subsequence, the family
$(\kappa_M, t_M)$ satisfies \eqref{ch2:eq:CWaty} but not
\eqref{ch2:eq:CWinf}, contradicting Theorem~\ref{th:CW}.
\end{remark}

\begin{proof}[Proof of Theorem~\ref{th:CW}]
Let $Y$ denote a random variable with characteristic function
\begin{equation} \label{ch2:eq:cfY}
	\E[e^{i u Y}] = e^{-|u|^\alpha\left(1+i \, \sign (u) \tan ( \frac{\pi \alpha}{2}) \right)} \,,
\end{equation}
i.e.\ $Y$ has a strictly stable law with index $\alpha$ 
and skewness parameter $\beta = -1$, and $\E[Y] = 0$.

If we set
\begin{equation}\label{ch2:eq:Yt}
	Y_t := \frac{X_t - \mu t}{\sigma t^{1/\alpha}} \,,
\end{equation}
it follows by \eqref{ch2:eq:charX} that $Y_t$ has
the same distribution as $Y$, because
\begin{equation} \label{ch2:eq:cfYt}
	\E[e^{i u Y_t}] =
	\E[e^{i u Y}] = e^{-|u|^\alpha\left(1+i \, \sign (u) \tan ( \frac{\pi \alpha}{2}) \right)} \,.
\end{equation}
It follows by \eqref{ch2:eq:Yt} that
\begin{equation} \label{ch2:eq:convXresc}
	\frac{X_t}{t^{1/\alpha}} 
	\xrightarrow[t\downarrow 0]{d} \sigma Y \,,
\end{equation}
hence Hypothesis~\ref{ch2:ass:smalltime} is satisfied
with $\gamma_t := t^{1/\alpha}$. 

Note that $\P(Y > a) > 0$ and $\P(Y < -a) > 0$ for all $a\in\R$, because
the density of $Y$ is strictly positive everywhere.
The right tail of $Y$ has a super-exponential decay: as $\kappa\to\infty$
\begin{equation} \label{ch2:eq:Y}
	\log \P(Y > \kappa) \sim -
	\tilde B_\alpha \, \kappa^{\alpha/(\alpha-1)}
	\qquad \text{where} \qquad
	\tilde B_\alpha := \frac{\alpha-1}{\alpha}
	\bigg( \frac{|\cos(\frac{\pi \alpha}{2})|}{\alpha} \bigg)^{1/(\alpha-1)} \,,
\end{equation}
cf.\ \cite[Property 1 and references therein]{cf:CW04}.
On the other hand the left tail is polynomial:
there exists $c = c_\alpha \in (0,\infty)$ such that
\begin{equation}\label{ch2:eq:Yleft}
       \P(Y \le -\kappa) \sim \frac{c}{\kappa^\alpha} \,,
\qquad \text{hence} \qquad
	\log \P(Y \le -\kappa) \sim -\alpha \log\kappa \,.
\end{equation}
Recalling that $\F_t(\kappa) := \P(X_t > \kappa)$ and $F_t(-\kappa) := \P(X_t \le \kappa)$,
by \eqref{ch2:eq:Yt} we can write
\begin{equation} \label{ch2:eq:FXY}
	\F_t(\kappa) = \P\bigg( Y > \frac{\kappa-\mu t}{\sigma t^{1/\alpha}} \bigg) \,,
	\qquad
	F_t(-\kappa) = \P\bigg( Y \le \frac{-\kappa-\mu t}{\sigma t^{1/\alpha}} \bigg) \,,
\end{equation}
hence we can transfer the estimates \eqref{ch2:eq:Y} and \eqref{ch2:eq:Yleft} to $X_t$.

Henceforth we consider separately the regimes of atypical deviations
\eqref{ch2:eq:CWaty}, and that of typical deviations \eqref{ch2:eq:CWty}.
Note that it is easy to check that \eqref{ch2:eq:qualef}
is equivalent to \eqref{ch2:eq:CWinfl}.

\subsubsection{Atypical deviations}

Let us fix an arbitrary family of values of
$(\kappa,t)$ satisfying \eqref{ch2:eq:CWaty}.
Then $\kappa / t \to \infty$ (because $\alpha > 1$), hence
\begin{equation*}
	\frac{\kappa-\mu t}{\sigma t^{1/\alpha}} \sim
	\frac{\kappa}{\sigma t^{1/\alpha}} \to \infty \,,
	\qquad
	\frac{-\kappa-\mu t}{\sigma t^{1/\alpha}} \sim
	\frac{-\kappa}{\sigma t^{1/\alpha}} \to -\infty \,.
\end{equation*}
By \eqref{ch2:eq:Y}, \eqref{ch2:eq:Yleft} and \eqref{ch2:eq:FXY} we then obtain
\begin{equation} \label{ch2:eq:CWF}
	\log\F_t(\kappa) \sim - \tilde B_\alpha \,
	\bigg( \frac{\kappa}{\sigma t^{1/\alpha}}
	\bigg)^{\alpha/(\alpha-1)} \,,
	\qquad
	\log F_t(-\kappa) \sim 
	- \log \frac{\kappa^\alpha}{t} \,.
\end{equation}

Let us now check the assumptions of Theorem~\ref{ch2:th:main2b}.
\begin{itemize}
\item The first relation in \eqref{ch2:eq:CWF} shows that Hypothesis~\ref{ch2:ass:rv} is satisfied
by the right tail $\F_t(\kappa)$, with $I_+(\rho) = \rho^{\alpha/(\alpha-1)}$.
Note that $I_+(\rho) \ge \rho$ for all $\rho\ge 1$, since $\alpha > 1$,
hence also condition \eqref{ch2:eq:Iplus} is satisfied.

\item Condition \eqref{ch2:eq:moment} is satisfied because
\eqref{ch2:eq:momentsimple} holds for all $T > 0$ and $\eta > 0$,
by \eqref{ch2:eq:momgenCW}.

\item It remains to check  condition \eqref{ch2:eq:moment0p}.
As we prove below, for all $\eta \in (0, \alpha-1)$ and
$T>0$ there are constants $A, B, C \in (0, \infty)$,
depending on $\eta, T$
and on the parameters $\alpha, \sigma$, such that for all $0 < t \le T$
and $\kappa \ge 0$ the following inequality holds:
\begin{equation} \label{ch2:eq:superest}
	\E \bigg[\bigg|\frac{e^{X_t} - 1}{\kappa} \bigg|^{1+\eta} \bigg] \le
        A \bigg( \bigg( \frac{t^{1/\alpha}}{\kappa} \bigg)^{B} + C \bigg) \,.
\end{equation}
In particular, since $\kappa / t^{1/\alpha} \to \infty$ by assumption \eqref{ch2:eq:CWaty},
condition \eqref{ch2:eq:moment0p} is satisfied.
\end{itemize}
Applying Theorem~\ref{ch2:th:main2b}, since
$-\log \F_t(\kappa)/\kappa \to \infty$ by the first relation in \eqref{ch2:eq:CWF},
the asymptotic behavior of $\sigma_\imp(\kappa,t)$
is given by \eqref{ch2:eq:sigmafin}, which
by \eqref{ch2:eq:CWF} coincides with \eqref{ch2:eq:CWinf}.

\smallskip

Next we want to apply Theorem~\ref{ch2:th:main2bl}.
By the second relation in \eqref{ch2:eq:CWF}, Hypothesis~\ref{ch2:ass:rv} is satisfied
by the left tail $F_t(-\kappa)$, with $I_-(\rho) \equiv 1$.
If $\kappa$ is bounded away from zero, the asymptotic behavior of $\sigma_\imp(\kappa,t)$
is given by \eqref{ch2:eq:sigmainfl}, which by \eqref{ch2:eq:CWF} yields
precisely \eqref{ch2:eq:qualef}.

If $\kappa \to 0$ we cannot apply directly Theorem~\ref{ch2:th:main2bl},
because the moment condition \eqref{ch2:eq:moment0p} is satisfied only for
some $\eta > 0$, and condition \eqref{ch2:eq:I-infty} is not satisfied, since
$I_-(\rho) \equiv 1$. However, we can show that \eqref{ch2:eq:ma2p} still holds by direct
estimates. By \eqref{ch2:eq:cp}
\begin{equation*}
	p(-\kappa,t) = \E[(e^{-\kappa} - e^{X_t}) \ind_{\{X_t < -\kappa\}}]
	\ge \E[(e^{-\kappa} - e^{X_t}) \ind_{\{X_t < -2\kappa\}}]
	\ge (e^{-\kappa} - e^{-2\kappa}) F_t(-2\kappa) \,,
\end{equation*}
and since $(e^{-\kappa} - e^{-2\kappa}) = e^{-2\kappa}(e^\kappa-1) \ge e^{-2\kappa}\kappa$,
we can write by \eqref{ch2:eq:CWF} (recall that $\kappa\to 0$)
\begin{equation} \label{ch2:eq:lbpCW}
	\log \big( p(-\kappa,t) / \kappa \big) \ge -2\kappa -
     \log \frac{(2\kappa)^\alpha}{t} \sim - \log \frac{\kappa^\alpha}{t}  \,.
\end{equation}
Next we give a matching upper bond on
$p(-\kappa,t)$. Since $\mu t \le \kappa$ eventually
(recall that $\kappa/t^{1/\alpha}\to\infty$, hence $\kappa/t\to \infty$), by \eqref{ch2:eq:FXY}
and \eqref{ch2:eq:Yleft} we obtain, for all $y\ge 1$
\begin{equation*}
	F_t ( -\kappa y) \le \P \bigg( Y \le
      - \frac{2\kappa y}{\sigma t^{1/\alpha}} \bigg)
      \le c' \frac{t}{\kappa^\alpha y^\alpha}\,,
\end{equation*}
for some $c' = c'_{\alpha,\sigma,\mu} \in (0,\infty)$. Then by Fubini's theorem
\begin{equation*}
\begin{split}
	p(-\kappa,t)& = \E[(e^{-\kappa} - e^{X_t}) \ind_{\{X_t < -\kappa\}}]
       = \E\bigg[ \int_\kappa^\infty
       e^{-x} \ind_{\{x < -X_t\}} \dd x \bigg] =
       \int_\kappa^\infty e^{-x} \, F_t (-x) \, \dd x \\
      & = \kappa \int_1^\infty e^{-\kappa y} \, F_t(-\kappa y) \, \dd y
       \le c' \kappa \,   \frac{t}{\kappa^\alpha}
       \int_1^\infty   \frac{1}{y^\alpha} \dd y =:
        c'' \kappa \, \frac{t}{\kappa^\alpha} \,,
\end{split}
\end{equation*}
hence
\begin{equation*}
	\log \big( p(-\kappa,t) / \kappa \big) \le \log c'' -
     \log \frac{\kappa^\alpha}{t} \sim - \log \frac{\kappa^\alpha}{t}  \,.
\end{equation*}
This relation, together with \eqref{ch2:eq:lbpCW}, yields
\begin{equation*}
	\log \big( p(-\kappa,t) / \kappa \big) 
       \sim - \log \frac{\kappa^\alpha}{t}  \,.
\end{equation*}
Since $\kappa / t^{1/\alpha} \to \infty$, this shows that
we are in the regime when $\kappa \to 0$ and $p(-\kappa,t) / \kappa \to 0$.
We can thus apply equation \eqref{ch2:eq:Vas<infty0} in
Theorem~\ref{ch2:th:main1}, which recalling Remark~\ref{ch2:rem:simple}
simplifies as the first line in \eqref{ch2:eq:Vas<infty0simple}
(with $p(-\kappa,t)$ instead of $c(\kappa,t)$), yielding
\begin{equation*}
	\sigma_\imp(-\kappa,t) \sim 
	\frac{\kappa}{\sqrt{2t\, ( - \log ( p(-\kappa,t) / \kappa ) )}} \sim
	\frac{\kappa}{\sqrt{2t\, \log \frac{\kappa^\alpha}{t}}} \,,
\end{equation*}
hence \eqref{ch2:eq:qualef} is proved also when $\kappa\to 0$
(cf.\ \eqref{ch2:eq:CWinfl}).

\subsubsection{Typical deviations.} 

Let us fix an arbitrary family of values of
$(\kappa,t)$ satisfying \eqref{ch2:eq:CWty}.
Relation \eqref{ch2:eq:superest} for $\kappa = \gamma_t = t^{1/\alpha}$
shows that condition \eqref{ch2:eq:assX} is satisfied, and Hypothesis~\ref{ch2:ass:smalltime}
holds by \eqref{ch2:eq:convXresc}. We can then apply Theorem~\ref{ch2:th:main2a},
and relation \eqref{ch2:eq:sigma0} gives precisely \eqref{ch2:eq:CWzero}.
\end{proof}

\begin{proof}[Proof of \eqref{ch2:eq:superest}]
Since $|\frac{e^x-1}{x}| \le 1$ if $x < 0$
and $|\frac{e^x-1}{x}| \le e^x$ if $x \ge 0$, we have
$|\frac{e^x-1}{x}| \le 1 + e^x$ for all $x\in\R$.
If $p, q > 1$ are such that $\frac{1}{p} + \frac{1}{q} = 1$,
Young's inequality $ab \le \frac{1}{p}a^p + \frac{1}{q} b^q$ yields
\begin{equation*}
\begin{split}
	\bigg|\frac{e^{X_t} - 1}{\kappa} \bigg| &=
	\bigg| \frac{X_t}{\kappa} \bigg| \, \bigg|\frac{e^{X_t} - 1}{X_t} \bigg|
	\le \frac{1}{p} \bigg( \frac{|X_t|}{\kappa} \bigg)^{p} +
	\frac{1}{q} \big( 1+ e^{X_t} \big)^{q} \,.
\end{split}
\end{equation*}
Noting that $(a+b)^r \le 2^{r-1}(a^r + b^r)$ for $r \ge 1$, by H\"older's inequality,
and denoting by $c = c_{p,\eta}$ a suitable constant depending only on $p, \eta$, we can write
\begin{equation*}
	\bigg|\frac{e^{X_t} - 1}{\kappa} \bigg|^{1+\eta} \le
        c \bigg( \frac{|X_t|^{p(1+\eta)}}{\kappa^{p(1+\eta)}} + 
        1 + e^{q(1+\eta)X_t} \bigg) \,.
\end{equation*}
Given $0 < \eta < \alpha - 1$, we fix 
$p = p_{\eta,\alpha} >1$ such that $B := p(1+\eta) < \alpha$.
(Note that $B$ depends only on $\eta, \alpha$.)
Moreover, it follows by \eqref{ch2:eq:Yt} that
\begin{equation*}
	\E[|X_t|^B] = (\sigma t^{1/\alpha})^B \, \E[|Y|^B] \, \big( 1+\cO(t^{B(1-1/\alpha)}) \big) \,,
\end{equation*}
and note that $\E[|Y|^B] < \infty$,
because $Y$ has finite moments of all orders strictly less than $\alpha$,
cf.\ \cite[Property 1]{cf:CW04}.
Since for $t \le T$ one has 
$\E[ e^{q(1+\eta)X_t} ] \le \E[ e^{q(1+\eta)X_T} ] < \infty$, 
by \eqref{ch2:eq:momgenCW}, relation \eqref{ch2:eq:superest} is proved.
\end{proof}

\subsection{Merton's Jump Diffusion Model} 
\label{ch2:sec:Merton}

Consider a model \cite{cf:M} where the log-return $X_t$ has an infinitely divisible distribution,
whose moment generating function is given by
\begin{equation}\label{ch2:eq:MertonGeneratingFunction}
 \E\left[\exp(zX_t)\right] = \exp \bigg( t \left\{z\mu+\frac{1}{2}z^2\sigma^2+
 \lambda\left(e^{z\alpha+z^2\frac{\delta^2}{2}}-1\right)\right\} \bigg) \,,
 \qquad \forall z \in \C \,,
\end{equation}
where $\mu, \alpha \in \R$ and $\sigma, \lambda, \delta \in (0,\infty)$
are fixed parameters.

The asymptotic behavior of $\sigma_\imp(\kappa,t)$ has
been studied by many authors.
The case of fixed $t>0$ and $\kappa \to \infty$
was derived by Benaim and Friz~\cite{cf:BF09}
using saddle point methods (for the detailed computation see \cite{cf:FGY14}, \cite{cf:GMZ14}).
The case of fixed $\kappa > 0$ and $t \to 0$
follows by \cite{cf:FF}, while the mixed regime of $t \to 0$, $\kappa \to 0$
with $\kappa \propto \sqrt{t \log(1/t)}$ was considered in \cite{cf:MT12}.
Applying our results, we can complete the picture, providing general formulas
which interpolate between all these regimes, cf.\ Theorem~\ref{ch2:th:mertonsmile}.

Let us define two functions $\bkappa_1(t) \to 0$ and $\bkappa_2(t) \to \infty$
as $t \to 0$ as follows:
\begin{equation}\label{eq:k12}
	\bkappa_1(t) := \sqrt{t \, \log \tfrac{1}{t}} \,, \qquad
	\bkappa_2(t) := \sqrt{\log \, \tfrac{1}{t}} \,,
\end{equation}
which will separate different behaviors as $t \to 0$. 
We stress that $\bkappa_1(t)$ is precisely the scaling considered in \cite{cf:MT12}.
Let us also define $f:[0,\infty) \to (0,\infty)$ by
\begin{equation}\label{eq:fa}
 f(a):= \min_{n \in \N} \left(n+\frac{a^2}{2n \delta^2}\right) 
\,.
\end{equation}
Note that $f$ is continuous and piecewise quadratic: 
more precisely, by explicit computation,
$f(a) = n + \frac{a^2}{2n\delta^2}$ for all 
$a \in [\sqrt{2(n-1)n} \,\delta, 
\sqrt{2n(n+1)} \, \delta)$, with $n\in\N$. It follows that
\begin{equation}\label{eq:asf}
	f(0) = 1 \,, \qquad
	f(a) \underset{\,a\to\infty\,}\sim \frac{\sqrt{2}}{\delta} \, a \,.
\end{equation}
The role of the function $f$ is explained by the following Lemma,
proved in Appendix~\ref{sec:lemma}.
\begin{lemma} \label{th:lemma}
For every fixed $a \in (0,\infty)$, as $t \to 0$ one has
\begin{equation}\label{eq:antici}
	\log \P(X_t > a \, \bkappa_2(t)) \,\sim\,
	-f(a) \log \frac 1t  \,.
\end{equation}
Moreover, if either $t \to 0$ and $\kappa \gg \bkappa_2(t)$, or if $t \to \bar t \in (0,\infty)$
and $\kappa \to \infty$, one has
\begin{equation}\label{eq:antici2}
	\log \P(X_t > \kappa) \sim - \frac{\kappa}{\delta}\sqrt{2\log \frac{\kappa}{t}} \,.
\end{equation}
\end{lemma}

We are now ready to state our main result for Merton's model (see Figure~\ref{ch2:fig:Merton}).

\begin{figure}[t]
\centering
\includegraphics[width=.5\columnwidth]{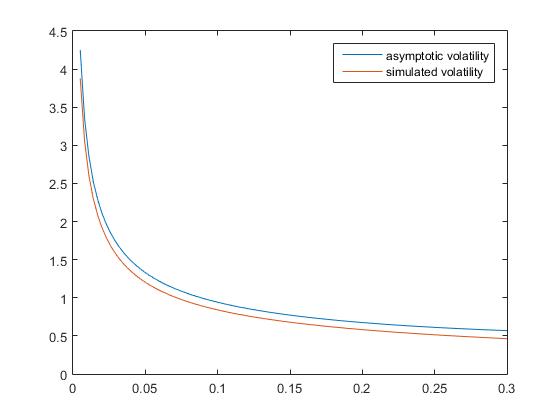}
\caption{Implied volatility $\sigma_\imp(\kappa,t)$
for Merton's model 
(with parameters $\lambda=0.01$, $\delta=0.3$, $\sigma = 0.2$, $\alpha=0.1$)
in the regime $\kappa = \bkappa_2(t)$ for $t \in (0,0.3)$.
The \emph{asymptotic volatility} is given by our formula \eqref{ch2:eq:mertontutto},
while the \emph{simulated volatility} is obtained using standard computational
packages.}
\label{ch2:fig:Merton}
\end{figure}

\begin{theorem}[Smile Asymptotics of Merton's model]\label{ch2:th:mertonsmile}
   Consider a family of values of $(\kappa,t)$ with $\kappa \geq 0$ and $t > 0$.
  \begin{enumerate}
   \item \label{svm:a} 
If $t \to 0$ and $\kappa = \cO(\bkappa_2(t))$, then
\begin{equation}\label{ch2:eq:mertontutto}
\sigma_\imp(\kappa,t) \sim \max \left\{ \sigma \,, \,
\frac{\kappa}{\sqrt{2t \, f\big( \frac{\kappa}{\bkappa_2(t)} \big) 
\log \frac{\kappa}{t}}} \right\}\,,
\end{equation}
which can be rewritten more explicitly as follows:
\begin{equation}\label{ch2:eq:mertontuttoexpl}
\displaystyle
\sigma_\imp(\kappa,t) \sim \begin{cases}
\sigma & \text{if } \ 0 \le \kappa \le \sigma \, \bkappa_1(t)\\
\displaystyle
\rule{0pt}{1.5em}
\frac{\kappa}{\sqrt{2t \log \frac{\kappa}{t}}} 
& \text{if } \ \sigma \, \bkappa_1(t) \le \kappa \ll \bkappa_2(t) \\
\displaystyle
\rule{0pt}{1.5em}
\frac{\kappa}{\sqrt{2t f(a) \log \frac{\kappa}{t}}}
& \text{if } \ \kappa \sim a\, \bkappa_2(t)
\ \text{ with } \ a \in (0,\infty) \\
\end{cases}
\,.
\end{equation}

\item\label{svm:e} If $\,t \to 0\,$ and $\,\kappa \gg \bkappa_2(t)\,$,
or if $\,\,t \to \bar t \in (0,\infty)\,$ and $\,\kappa \to \infty\,$, then
\begin{equation}\label{ch2:eq:mertonvole}
\sigma_\imp(\kappa,t) \sim 
\sqrt{\frac{\delta \, \kappa}{2t \sqrt{2\log\frac{\kappa}{t}}}} \,.
\end{equation}
  \end{enumerate}
\end{theorem}


\begin{proof}

We have to prove relation \eqref{ch2:eq:mertontuttoexpl}
(which is equivalent to \eqref{ch2:eq:mertontutto}, by extracting subsequences)
and relation \eqref{ch2:eq:mertonvole}.
We distinguish different subcases.

\medskip

Assume first that $t\to 0$ with $\kappa = \cO(\sqrt{t})$. By extracting sub-sequences, assume that
$\kappa /\sqrt{t} \to a \in [0,\infty)$.
Note that $X_t/\sqrt{t} \xrightarrow[]{\, d\, } N(0, \sigma^2)$, 
because $\E[e^{i u \frac{X_t}{\sqrt{t}}}] \to e^{-\frac{u^2}{2\sigma^2}}$ for every $u\in\R$,
as one checks by \eqref{ch2:eq:MertonGeneratingFunction}.
We then apply 
Theorem~\ref{ch2:th:main2a} with $\gamma_t = \sqrt{t}$, 
because the moment condition \eqref{ch2:eq:assX} for $\eta = 1$ 
follows by \eqref{ch2:eq:MertonGeneratingFunction} (see 
also \eqref{ch2:eq:moment0psimple}).
Relation \eqref{ch2:eq:sigma0} yields
\begin{equation*}
	\sigma(\kappa,t) \sim C_+(a) \, \frac{\sqrt{t}}{\sqrt{t}} = \sigma \,,
\end{equation*}
since $C_+(a) \equiv \sigma$ (see \eqref{eq:DD-1} in Appendix~\ref{app:diffusions}).
This matches with the first line of \eqref{ch2:eq:mertontuttoexpl}.

\medskip

Next we assume that $t\to 0$ with $\sqrt{t} \ll \kappa \ll 1$.
Applying \cite[Proposition 2.3]{cf:MT12}, we can write
\begin{equation*} 
 c(\kappa,t) \sim \E[(e^{\sigma W_t-\frac{\sigma^2 t}{2}}-e^\kappa)^+]
 \,+\, C \, t \,, 
\end{equation*}
with $0 < C := \int_0^{\infty}(e^x-1)\nu(\dd x) < \infty$, where
$\nu$ denotes the L\'evy measure of $X$. The first term 
is the usual Black\&Scholes price of a call option with $\kappa \gg \sqrt{t}$: applying
\eqref{ch2:eq:asv} below with
$v = \sigma \sqrt{t}$ and $d_1 = -\frac{\kappa}{v}+\frac{v}{2} \sim -\frac{\kappa}{\sigma\sqrt{t}}$,
together with \eqref{ch2:eq:phiPhi} and \eqref{ch2:eq:Millsas}, we get
\begin{equation*}
\begin{split}
 \frac{c(\kappa,t)}{\kappa} & \sim \frac{v}{\kappa} \, \frac{\phi(d_1)}{(d_1)^2}
 \,+\, C \, \frac{t}{\kappa} \,=\,
 \frac{e^{-\frac{d_1^2}{2}}}
 {\sqrt{2\pi} \, (\frac{\kappa^2}{\sigma^2 t})^{3/2}} \,+\, C \, \frac{t}{\kappa}
\, \sim \,  \frac{e^{-\frac{d_1^2}{2}}}{\sqrt{2\pi} \, d_1^3}
 \,+\, C \, \frac{t}{\kappa} \,.
\end{split}
\end{equation*}
Writing $\frac{e^{-\frac{z^2}{2}}}{\sqrt{2\pi} z^3} = 
e^{-\frac{z^2}{2} - \log(\sqrt{2\pi}z^3)}
= e^{-\frac{z^2}{2}(1+o(1))}$
as $z \to \infty$, we get
\begin{equation*}
 \frac{c(\kappa,t)}{\kappa} \sim
 e^{-\frac{\kappa^2}{2\sigma^2 t} (1+o(1))}
 \,+\, C \, \frac{t}{\kappa} = a + b \quad \text{(say)} \,.
\end{equation*}
The inequalities $\max\{a,b\} \le a+b \le 2 \max\{a,b\}$ yield
$\log(a+b) \sim \max\{\log a, \log b\}$ (the additive constant $\log 2$ is irrelevant,
since $a,b \to 0$) hence
\begin{equation*}
 - \log \frac{c(\kappa,t)}{\kappa} \sim
 -\max \bigg\{ -\frac{\kappa^2}{2\sigma^2 t}(1+o(1)), \ \log \bigg( C \frac{t}{\kappa}
 \bigg) \bigg\}
 \sim
 \min \bigg\{ \frac{\kappa^2}{2\sigma^2 t}, \ \log \frac{\kappa}{t} \bigg\} \,.
\end{equation*}
It is easy to check that 
the asymptotic equality $\frac{\kappa^2}{2\sigma^2 t} \sim \log \frac{\kappa}{t}$ 
holds when $\kappa \sim \sigma \bkappa_1(t)$.
It follows that:
\begin{itemize}
\item in case $\kappa \le \sigma \bkappa_1(t)$ we have
\begin{equation} \label{eq:plu1}
	-\log \frac{c(\kappa,t)}{\kappa} \sim \frac{\kappa^2}{2\sigma^2 t} \,;
\end{equation}

\item in case $\kappa \ge \sigma \bkappa_1(t)$ we have
\begin{equation} \label{eq:plu2}
	-\log \frac{c(\kappa,t)}{\kappa} \sim \log \frac{\kappa}{t} \,.
\end{equation}
\end{itemize}
(Note that when $\kappa \sim \sigma \bkappa_1(t)$
both relations \eqref{eq:plu1} and \eqref{eq:plu2} hold.)

We can deduce the asymptotic behavior of $\sigma_\imp(\kappa,t)$
applying relation \eqref{ch2:eq:Vas<infty0} 
(note that $c(\kappa,t)/\kappa \to 0$, since $\kappa \gg \sqrt{t}$)
which, by Remark~\ref{ch2:rem:simple},
reduces to the first line of \eqref{ch2:eq:Vas<infty0simple}, i.e.\
\begin{equation} \label{eq:sigmafina}
	\sigma_\imp(\kappa,t) \sim \frac{\kappa}{\sqrt{2t\, ( - \log ( c(\kappa,t) / \kappa ) )}} \,.
\end{equation}
Plugging \eqref{eq:plu1}-\eqref{eq:plu2} into this relation yields
the first and second line of \eqref{ch2:eq:mertontuttoexpl}.

\medskip

Next we assume that
$t \to 0$ with $\eta \le \kappa \ll \bkappa_2(t)$, for some fixed $\eta > 0$.
We claim that
\begin{equation}\label{eq:wecla}
	-\log \frac{c(\kappa,t)}{\kappa} \sim \log \frac{1}{t} \,,
\end{equation}
which plugged into \eqref{eq:sigmafina} proves the second line of \eqref{ch2:eq:mertontuttoexpl}
(since $\log \frac{1}{t} \sim \log \frac{\kappa}{t}$ in this regime).
When $\kappa > 0$ is fixed, (a stronger version of)
relation \eqref{eq:wecla} was proved by Figueroa--L\'{o}pez and Forde
in \cite{cf:FF}. For the general case, we fix $a \in (0,\infty)$
and we apply relation \eqref{eq:antici}, which yields the lower bound
\begin{equation} \label{eq:lobb}
\begin{split}
	c(\kappa,t) = \E[(e^{X_t}-e^\kappa)^+] 
	\ge (e^{a \, \bkappa_2(t)}-e^{\kappa}) \, \P(X_t >a \, \bkappa_2(t))
	& \sim e^{a \, \bkappa_2(t) - f(a) \log\frac{1}{t}(1+o(1))} \\
	& \sim e^{- f(a) \log\frac{1}{t}(1+o(1))} \,,
\end{split}
\end{equation}
because $\bkappa_2(t) = \sqrt{\log\frac{1}{t}} 
\ll \log \frac{1}{t}$. For an upper bound, we recall that $\kappa \ge \eta$
and \cite{cf:FH}
\begin{equation}\label{eq:showbe}
	\log \P(X_t > \eta)  \, \sim \, - \log \frac{1}{t}
	\qquad \text{as } t \to 0 \,.
\end{equation}
Then, for every fixed $b \in (0,\infty)$,
using \eqref{eq:showbe}, \eqref{eq:antici} and Caucy-Schwarz
we can write
\begin{equation*}
\begin{split}
	c(\kappa,t) & 
	\le e^{b \, \bkappa_2(t)} \, \P(\eta < X_t \le b \, \bkappa_2(t))
	+ \E[e^{X_t} \, \ind_{\{X_t > b \, \bkappa_2(t)\}}] \\
	& \le e^{b \, \bkappa_2(t) - \log \frac{1}{t}(1+o(1))}
	+ \E[e^{2X_t}]^{\frac{1}{2}} \, \P(X_t > b \, \bkappa_2(t))^{\frac{1}{2}} \\
	& \sim e^{- \log \frac{1}{t}(1+o(1))} +
	C_1 \, e^{- \frac{1}{2}f(b) \log \frac{1}{t}(1+o(1))}\,,
\end{split}
\end{equation*}
where in the last step we used $\E[e^{2X_t}]^{\frac{1}{2}} \le C_1$ for some 
constant $C_1$,
since $\E[e^{2X_t}] \to 1$ as $t \to 0$ by \eqref{ch2:eq:MertonGeneratingFunction}.
By \eqref{eq:asf}, we can fix $b$ large enough so that $f(b) > 2$,
hence $c(\kappa,t) \le e^{- \log \frac{1}{t}(1+o(1))}$,
and for every $\epsilon > 0$ we can choose $a > 0$ small enough
such that $f(a) < 1+\epsilon$,
hence $c(\kappa,t) \ge e^{-(1+\epsilon) \log \frac{1}{t}(1+o(1))}$ 
by \eqref{eq:lobb}. Altogether, relation \eqref{eq:wecla} is proved.

\medskip

Let us proceed with
the regime $t \to 0$ and $\kappa \sim a \, \bkappa_2(t)$ for some $a \in (0,\infty)$.
Relation \eqref{eq:antici} shows that
such a family of $(\kappa,t)$ 
satisfies Hypothesis~\ref{ch2:ass:rv}
with $I_+(\rho) = f(\rho a) / f(a)$ (we stress that $a \in (0,\infty)$ is fixed
throughout this argument, hence $I_+$ can depend on $a$).
Since the moment condition \eqref{ch2:eq:moment} is clearly
satisfied by \eqref{ch2:eq:MertonGeneratingFunction},
we can apply Theorem~\ref{ch2:th:main2b}:
relation \eqref{ch2:eq:sigmafin}, coupled with \eqref{eq:antici},
proves the third line of \eqref{ch2:eq:mertontuttoexpl}.

\medskip

Finally, it remains to prove \eqref{ch2:eq:mertonvole}, hence we assume that
either $t \to 0$ and $\kappa \gg \bkappa_2(t)$, or $t \to \bar t \in (0,\infty)$
and $\kappa \to \infty$. 
Relation \eqref{eq:antici2} shows that
Hypothesis~\ref{ch2:ass:rv} holds with
$I_+(\rho) =\rho$. By Theorem~\ref{ch2:th:main2b},
relation \eqref{ch2:eq:sigmafin} yields 
\eqref{ch2:eq:mertonvole}, completing
the proof of Theorem~\ref{ch2:th:mertonsmile}.
\end{proof}

\subsection{The Heston Model} 
\label{ch2:sec:Heston}

Given the parameters $\lambda, \theta, \eta, \sigma_0 \in (0,\infty)$
and $\rho \in [-1,1]$,
the Heston model \cite{cf:H} is a stochastic volatility 
model $(S_t)_{t\ge 0}$ defined by the following SDEs
\begin{equation*}
\begin{cases}
\dd S_t\, = \,S_t\,\sqrt{V_t}\, \dd W^1_t \,, \\
\dd V_t\, = \, -\lambda(V_t-\theta)\, \dd t \,+\, \eta \sqrt{V_t} \, \dd W_t^2 \,, \\
X_0=0 \,, \quad V_0=\sigma_0 \,,
\end{cases}
\end{equation*}
where $(W^1_t)_{t\ge 0}$ and $(W^2_t)_{t\ge 0}$
are standard Brownian motions with
$\langle \dd W^1_t\, , \dd W^2_t \rangle = \rho\, \dd t$.

Note that $S_t$  displays explosion of moments, i.e.\
$\E[S_T^p] = \infty$ for $p > 1$ large enough.
In general, for any fixed $t\geq 0$ one can define the explosion moment $p^*(t)$ as
\[
 p^*(t) := \sup\{p > 0: \ \E[S_t^p] < \infty\}\,,
\]
so that $\E[S_t^p] < \infty$ for $p < p^*(t)$ while
$\E[S_t^p] = \infty$ for $p > p^*(t)$
(in the case of Heston's model, one has $\E[S_t^p] = \infty$ also for $p = p^*(t)$).
The behavior of the explosion moment $p^*(t)$ is described
in the following Lemma, proved below.

\begin{lemma}\label{ch2:th:lemmaheston}
If $\rho=-1$, then $p^*(t)=+\infty$ for every $t \geq 0$.\\
If $\rho>-1$, then $p^*(t) \in (1,+\infty)$ for every $t >0$. 
Moreover, as $t \downarrow 0$
 \[
  p^*(t) \sim \frac{C}{t} \,,
 \]
where
\begin{equation}\label{ch2:eq:CHeston}
C = C(\rho, \eta ) :=\begin{cases}
\displaystyle
\frac{2}{\eta \sqrt{1-\rho^2}}
\left(\arctan \frac{\sqrt{1-\rho^2}}{\rho }+\pi 1_{\rho<0}\right)
& \text{ if } \rho<1\\
\displaystyle
\rule{0pt}{2.2em}
\frac{2}{\eta} & \text{ if } \rho=1           
\end{cases} \,.
\end{equation}
\end{lemma}

The asymptotic behavior of the implied volatility
$\sigma_\imp(\kappa,t)$ is known in the regimes
of large strike (with fixed maturity) and small maturity (with fixed strike).
\begin{itemize}
\item In \cite{cf:BF08}, Benaim and Friz show that
for fixed $t > 0$, when $\kappa \to +\infty$ 
\begin{equation}\label{ch2:eq:impvolHestonk}
	\sigma_\imp(\kappa, t)
	\underset{\kappa \uparrow \infty}{\sim} 
	\frac{\sqrt{2\kappa}}{\sqrt{ t}} \left(\sqrt{p^*(t)}-\sqrt{p^*(t)-1}\right)\,,
\end{equation}
based on the estimate (cf.\ also \cite{cf:AP07})
\begin{equation}\label{ch2:eq:estH1}
 -\log \P(X_t>\kappa)
	\underset{\kappa \uparrow \infty}{\sim} 
  p^*(t) \, \kappa\,.
\end{equation}

\item In \cite{cf:FJ09}, Forde and Jacquier have proved that
for any fixed $\kappa > 0$, as $t \downarrow 0$
\begin{equation}\label{ch2:eq:impvolHestont}
	\sigma_\imp(\kappa,t) 
	\underset{t \downarrow 0}{\sim} 
	\frac{\kappa}{\sqrt{2\, \Lambda^*(\kappa)}} \,,
\end{equation}
where $\Lambda^*(\cdot)$ is the 
Legendre transform of the function $\Lambda :\R^+\to \R^+ \cup\{\infty\}$ given by
\begin{equation}
\Lambda(p) :=
\begin{cases}
\displaystyle
\frac{\sigma_0 p}{\eta\left(\sqrt{1-\rho^2}\cot\left(\frac{1}{2}\eta 
p\sqrt{1-\rho^2}\right)-\rho \right)} & \text{if } p < C \,, \\
\rule{0pt}{1.4em}
\infty & \text{if } p \geq C \,,
\end{cases}
\end{equation}
where $C$ is the constant in \eqref{ch2:eq:CHeston}.
Their analysis is based on the estimate
\begin{equation} \label{ch2:eq:estH2}
	-\log \P(X_t \ge \kappa) 
	\underset{t \downarrow 0}{\sim} 
	 \frac{1}{t} \Lambda^*(\kappa) \,,
\end{equation}
obtained by showing that the log-price $(X_t)_{t\ge 0}$ in the Heston model satisfies 
a large deviations principle as $t \downarrow 0$, with rate $1/t$ and good 
rate function $\Lambda^*(\kappa)$.
\end{itemize}

We first note that the asymptotics \eqref{ch2:eq:impvolHestonk} and
\eqref{ch2:eq:impvolHestont} follow easily from our Theorem~\ref{ch2:th:main2b},
plugging the estimates \eqref{ch2:eq:estH1} and \eqref{ch2:eq:estH2}
into relations \eqref{ch2:eq:sigmainf} and \eqref{ch2:eq:sigmafin}, respectively.

\smallskip

We also observe that the estimates \eqref{ch2:eq:impvolHestonk} and
\eqref{ch2:eq:impvolHestont} match, in the following sense:
if we take the limit $t \to 0$ of the right hand side of \eqref{ch2:eq:impvolHestonk}
(i.e.\ we first let $\kappa \uparrow + \infty$ and then $t\downarrow 0$
in $\sigma_\imp(\kappa, t)$), we obtain
\begin{equation} \label{ch2:eq:coinci}
\eqref{ch2:eq:impvolHestonk}
\underset{t \downarrow 0}{\sim} 	
	\frac{\sqrt{2\kappa}}{\sqrt{ t}} \frac{1}{2 \sqrt{p^*(t)}}
\sim \frac{\sqrt{2\kappa}}{\sqrt{ t}} \frac{1}{2 \sqrt{\frac{C}{t}}}
= \frac{\sqrt{\kappa}}{\sqrt{2\, C}} \,.
\end{equation}
If, on the other hand, we take the limit $\kappa \uparrow \infty$ 
of the right hand side of \eqref{ch2:eq:impvolHestont}
(i.e.\ we first let $t\downarrow 0$ and then  $\kappa \uparrow + \infty$
in $\sigma_\imp(\kappa, t)$),
since $\Lambda^*(\kappa)\sim C \kappa$,\footnote{This is because 
$\Lambda(p) \uparrow +\infty$ as $p \uparrow C$, hence the slope of $\Lambda^*(\kappa)$
converges to $C$ as $\kappa \to \infty$.} we obtain
\begin{equation}
 \eqref{ch2:eq:impvolHestont}
 \underset{\kappa \uparrow +\infty}{\sim}
 \frac{\kappa}{\sqrt{2 C \kappa}}= \frac{\sqrt{\kappa}}{\sqrt{2 \,C}} \,,
\end{equation}
which coincides with \eqref{ch2:eq:coinci}.
Analogously, also the estimates \eqref{ch2:eq:estH1} and \eqref{ch2:eq:estH2} match.

It is then natural to conjecture that,
for \emph{any} family of values of $(\kappa,t)$ such
that $\kappa \uparrow +\infty$ and $t \downarrow 0$ jointly,
one should have
\begin{equation} \label{ch2:eq:conjH}
\log \P(X_t\geq \kappa) \sim -C \,\frac{\kappa}{t} \,,
\end{equation}
where $C$ is the constant in \eqref{ch2:eq:CHeston}.
If this holds, applying Theorem~\ref{ch2:th:main2b}, 
relation \eqref{ch2:eq:sigmafin} yields
\begin{equation} \label{ch2:eq:Hetint}
 \sigma_\imp(\kappa,t) \sim \frac{\sqrt{\kappa}}{\sqrt{2\,C}} \,,
\end{equation}
providing a smooth interpolation between \eqref{ch2:eq:impvolHestonk} and
\eqref{ch2:eq:impvolHestont}.

\begin{remark}[Surface asymptotics for the Heston model]\rm \label{ch2:rem:jointH}
If \eqref{ch2:eq:Hetint} holds for any family of values of $(\kappa,t)$
with $\kappa \to \infty$ and $t \to 0$, it follows that for every $\epsilon > 0$
there exists $M = M(\epsilon) \in (0,\infty)$ such that 
the following inequalities hold:
\begin{equation*}
	\big(1-\epsilon\big) \frac{\sqrt{\kappa}}{\sqrt{2\,C}} \le
	\sigma_\imp(\kappa,t) \le \big(1+\epsilon\big)\frac{\sqrt{\kappa}}{\sqrt{2\,C}} \,,
\end{equation*}
\emph{for all $(\kappa,t)$ in the region $\cA_{T,M} := \{0 < t \le \frac{1}{M}, \ \kappa > M\}$},
as it follows easily by contradiction (cf.\ Remark~\ref{ch2:rem:joint}
for a similar argument).
\end{remark}

\begin{proof}[Proof of Lemma \ref{ch2:th:lemmaheston}]
Given any number $p>1$ we define the explosion time $T^*(p)$ as
\begin{equation*}
	T^*(p) := \sup\{t > 0: \ \E[S_t^p] < \infty \} \,.
\end{equation*}
Note that if $T^*(p)=t \in (0,+\infty)$ then $p^*(t)=p$.
By \cite{cf:AP07} (see also \cite{cf:FK09})
\begin{equation}\label{ch2:eq:TstarHeston}
T^*(p)= \begin{cases} 
+\infty 
& \text{if } \
\Delta(p)\geq 0,\,\chi(p)<0 \,, \\
\rule{0pt}{1.8em}
\frac{1}{\sqrt{\Delta(p)}} \log \left(\frac{\chi(p)+\sqrt{\Delta (p)}}
{\chi(p)-\sqrt{\Delta (p)}}\right)
& \text{if } \
\Delta(p)\geq 0,\,\chi(p)>0 \,, \\
\rule{0pt}{2.0em}
\frac{2}{\sqrt{-\Delta (p)}}\left(\arctan \frac{\sqrt{-\Delta(p)}}{\chi(p)}
+\pi 1_{\chi(p)<0}\right) & \text{if } \ \Delta(p)<0 \,,
\end{cases}
\end{equation}
where
\begin{equation*}
	\chi(p) := \rho \eta p\,-\, \lambda\,,
	\qquad
	\Delta(p) := \chi^2(p)-\eta^2(p^2-p) \,,
\end{equation*}
Observe that if $\rho=-1$, then $\chi(p)=-\eta p-\lambda<0$ and 
$\Delta(p)=\lambda^2+p\left(2\eta\lambda+\eta^2\right)\geq0$, which implies 
$T^*(p)=+\infty $ for every $p>1$, or equivalently $p^*(t)=+\infty$ for every $t > 0$.

On the other hand, since
\[
\Delta(p)\,=\, \rho^2\eta^2p^2+\lambda^2-2\eta \rho\lambda p-\eta^2p^2+\eta^2 p \, =\,
\eta^2p^2(\rho^2-1)+p(\eta^2-2\eta\rho \lambda)+\lambda^2 \,,
\]
we observe that if $\rho \neq 1$, then $\Delta p<0$ as $p\to +\infty$, which implies
\begin{equation}
 \begin{split}
  T^*(p) &\underset{p \uparrow \infty}{\sim}\frac{2}{p(\eta\sqrt{1-\rho^2})}
\left(\arctan \frac{\eta p\sqrt{1-\rho^2}}{\rho \eta p}+\pi 1_{\rho<0}\right)\\
&=\, \frac{1}{p}\frac{2}{\eta \sqrt{1-\rho^2}}
\left(\arctan \frac{\sqrt{1-\rho^2}}{\rho }+\pi 1_{\rho<0}\right)\,.
 \end{split}
\end{equation}
In particular this leads to the conclusion that, if $|\rho|\neq1$, then
\[
p^*(t)\underset{t \downarrow 0}{\sim}\frac{C}{t}
\]
where $C$ was defined in \eqref{ch2:eq:CHeston}.

It remains to study the case $\rho =1$, in which $\chi(p)>0$ for every $p$. 
We have two possibilities:
if $\eta>2 \lambda$ then $\Delta(p)>0$ when $p \to +\infty$, and so by \eqref{ch2:eq:TstarHeston}
\begin{equation*}
 \begin{split}
 T^*(p) &\underset{p \uparrow \infty}{\sim} \frac{1}{\sqrt{p(\eta^2+2\eta \lambda)}}
 \log \left(1+2\frac{\sqrt{p(\eta^2+2\eta \lambda)}}{\eta p-\sqrt{p(\eta^2+2\eta \lambda)}}\right)
 \sim \frac{2}{\eta}\, \frac{1}{p} \,.
 \end{split}
\end{equation*}
On the other hand, if $\eta < 2 \lambda$,
then $\Delta(p) < 0$ when $p \to \infty$ and so
\begin{equation*}
 \begin{split}
  T^*(p) &\underset{p \uparrow \infty}{\sim}\frac{2}{\sqrt{p(2\eta\lambda-\eta^2)}}
\left(\arctan \frac{\sqrt{p(2\eta\lambda-\eta^2)}}{p\eta}\right)
\sim  \frac{2}{\eta}\, \frac{1}{p}\,.
 \end{split}
\end{equation*}
Finally if $\eta=2\lambda$, $\Delta(p)=\lambda^2$, and so
\[
T^*(p)=\frac{1}{\lambda}\log \left(1+\frac{2\lambda}{\eta p-2\lambda}\right) 
\underset{p\uparrow \infty}{\sim}  \frac{2}{\eta}\, \frac{1}{p} \,.
\]
In all the cases we obtain $p^*(t)\underset{t \downarrow 0}{\sim} \frac{2}{\eta} \, \frac{1}{t}$,
in agreement with \eqref{ch2:eq:CHeston}.
\end{proof}

\section{From option price to implied volatility}
\label{ch2:sec:pricetovol}

In this section we prove Theorem~\ref{ch2:th:main1}.
We start with some background on Black\&Scholes model
and on related quantities.
We let $Z$ be a standard Gaussian random variable and denote by
$\phi$ and $\Phi$ its density and distribution functions:
\begin{equation} \label{ch2:eq:phiPhi}
	\phi(z) := \frac{\P(Z \in \dd z)}{\dd z} = \frac{e^{-\frac{1}{2} z^2}}{\sqrt{2\pi}} \,, \qquad
	\Phi(z) := \P(Z \le z) = \int_{-\infty}^z \phi(t) \, \dd t \, .
\end{equation}

\subsection{Mills ratio}
\label{ch2:eq:backg}

The Mills ratio $U: \R \to (0,\infty)$ is defined by
\begin{equation}\label{ch2:eq:Mills}
	U(z) := \frac{1-\Phi(z)}{\phi(z)} = \frac{\Phi(-z)}{\phi(z)}, \qquad \forall z \in \R \,.
\end{equation}
The next lemma summarizes
the main properties of $U$ that will be used in the sequel.

\begin{lemma}\label{ch2:th:Mills}
The function $U$ is smooth, strictly decreasing, strictly convex and  satisfies
\begin{equation}\label{ch2:eq:Millsas}
	U'(z) \sim -\frac{1}{z^2} \qquad \text{as } z \uparrow \infty \,.
\end{equation}
\end{lemma}

\begin{proof}
Since $\Phi'(z) = \phi(z)$ and $\phi$ is an analytic function, $U$ is also analytic.
Since $\phi'(z) = -z\phi(z)$, one obtains
\begin{equation} \label{ch2:eq:U12}
	U'(z) = z U(z) - 1 \,, \qquad
	U''(z) = U(z) + zU'(z) = (1+ z^2) U(z) - z \,.
\end{equation}
Recalling that $U(z) > 0$, these relations already show that $U'(z) < 0$ and
$U''(z) > 0$ for all $z \le 0$. For $z > 0$, the following bounds hold
\cite[eq. (19)]{cf:S54}, \cite[Th. 1.5]{cf:P01}:
\begin{equation} \label{ch2:eq:Millsbounds}
	\frac{z}{z^2+1} =
	\frac{1}{z+\frac{1}{z}} < U(z) <
	\frac{1}{z+\frac{1}{z + \frac{2}{z}}} 
	= \frac{z^2 + 2}{z^3 + 3z} , 
	\qquad \forall z > 0 \,.
\end{equation}
Applying \eqref{ch2:eq:U12} yields $U''(z) > 0$ and
$-\frac{1}{1+z^2} < U'(z) < -\frac{1}{3+z^2}$ for all $z > 0$, hence \eqref{ch2:eq:Millsas}.
\end{proof}

We recall that the smooth function $D: (0,\infty) \to (0,\infty)$ was introduced in \eqref{ch2:eq:D}.
Since 
\begin{equation}\label{ch2:eq:D'}
D'(z) = -\frac{1}{z^2} \phi(z) < 0 \,,
\end{equation}
$D(\cdot)$ is a strictly decreasing bijection (note that 
$\lim_{z \downarrow 0} D(z) = \infty$ and
$\lim_{z \to\infty} D(z) = 0$). 
Its inverse $D^{-1}: (0,\infty) \to (0,\infty)$
is then smooth and strictly decreasing as well.
Writing
$D(z) = \phi(z) ( \frac{1}{z} - U(z) )$, it follows by \eqref{ch2:eq:Millsbounds} that
$\frac{1}{z} - U(z) \sim \frac{1}{z^3}$ as $z \uparrow \infty$, hence
\begin{equation*}
	D(z) \sim \frac{1}{z^3} \phi(z) \sim
	\frac{e^{-\frac{1}{2}z^2}}{\sqrt{2\pi} \, z^3}
	\quad \
	\text{as } z \uparrow \infty \,, \qquad \ 
	D(z) \sim \frac{1}{z} \phi(0) = \frac{1}{\sqrt{2\pi} z}
	\quad \
	\text{as } z \downarrow 0  \,.
\end{equation*}
It follows easily that $D^{-1}(\cdot)$ satisfies \eqref{ch2:eq:Das}.

\subsection{Black\&Scholes}
\label{ch2:sec:BS}

Let $(B_t)_{t\ge 0}$
be a standard Brownian motion.
The Black\&Scholes model is defined by a risk-neutral log-price
$(X_t := \sigma B_t - \frac{1}{2} \sigma^2 t)_{t\ge 0}$, where
the parameter $\sigma \in (0,\infty)$
represents the volatility.
The Black\&Scholes formula for the price of a normalized European call
is $\CBS(\kappa,\sigma \sqrt{t})$,
where $\kappa$ is the log-strike, $t$ is the maturity and we define
\begin{equation} \label{ch2:eq:BS1}
	\CBS(\kappa,v) :=  \E [ (e^{v Z - \frac{1}{2} v^2} - e^\kappa)^+ ]
	= \begin{cases}
	(1-e^\kappa)^+ & \text{if } v = 0 \,, \\
	\rule{0pt}{1.2em}\Phi(d_1) - e^\kappa \Phi(d_2)  & \text{if } v > 0 \,,
	\end{cases}
\end{equation}
where $\Phi$ is defined in \eqref{ch2:eq:phiPhi}, and we set
\begin{equation} \label{ch2:eq:BS2}
	\begin{cases}
	d_1 = d_1(\kappa,v) := - \frac{\kappa}{v} + \frac{v}{2} \,, \\
	d_2 = d_2(\kappa, v) := - \frac{\kappa}{v} - \frac{v}{2} \,,
	\end{cases}
	\quad \ \text{so that} \quad \
	\begin{cases}
	d_2 = d_1 - v \,, \\
	d_2^2 = d_1^2 + 2\kappa \,.
	\end{cases}
\end{equation}

Note that $\CBS(\kappa,v)$ is a continuous
function of $(\kappa,v)$. Since $e^\kappa \phi(d_2) = \phi(d_1)$,
for all $v > 0$ one easily computes
\begin{gather*}
	\frac{\partial \CBS(\kappa,v)}{\partial v}
	= \phi(d_1) > 0 \,, \qquad
	\frac{\partial \CBS(\kappa,v)}{\partial \kappa}
	= - e^\kappa \Phi(d_2) < 0 \,,
\end{gather*}
hence $\CBS(\kappa,v)$ is strictly increasing in $v$ and strictly
decreasing in $\kappa$ (see Figure~\ref{ch2:fig:BS}).
It is also directly checked that for all $\kappa \in \R$ and $v \ge 0$ one has
\begin{equation} \label{ch2:eq:CBS-}
	\CBS(\kappa,v) =  1 - e^\kappa + e^\kappa\CBS(-\kappa,v) \,.
\end{equation}

\smallskip

In the following key proposition,
proved in Appendix~\ref{ch2:sec:app:BS},
we show that when $\kappa \ge 0$
the Black\&Scholes call price $\CBS(\kappa, v)$ vanishes
precisely when $v \to 0$ or $d_1 \to -\infty$ (or, more generally, in a combination of
these two regimes, when $\min\{d_1, \log v\} \to -\infty$).
We also provide
sharp estimates on $\CBS(\kappa, v)$ for each regime
(weaker estimates on $\log \CBS(\kappa, v)$
could be deduced from Theorems~\ref{ch2:th:main2b}
and~\ref{ch2:th:main2bl}).

\begin{proposition}\label{ch2:th:BS}
For any family of values of $(\kappa,v)$ with $\kappa \ge 0$, $v > 0$,
one has
\begin{equation}\label{ch2:eq:cto0}
\CBS(\kappa, v)\to 0 \qquad \text{if and only if}
\qquad \min\{d_1, \log v\} \to -\infty \,,
\end{equation}
that is, $\CBS(\kappa, v)\to 0$ if and only if from any subsequence of $(\kappa,v)$ one can extract 
a sub-subsequence along which either
$d_1 \to -\infty$ or $v \to 0$. Moreover:
\begin{itemize}
\item if $d_1 := - \frac{\kappa}{v} + \frac{v}{2} \to -\infty$, then 
\begin{equation}\label{ch2:eq:asd1}
	\CBS(\kappa, v) \sim \phi(d_1) \, \frac{v}{-d_1(-d_1+v)} \,;
\end{equation}

\item if $v \to 0$, then
\begin{equation}\label{ch2:eq:asv}
	\CBS(\kappa, v) \sim - U'(-d_1) \,\phi(d_1) \, v \,;
\end{equation}
\end{itemize}
where $\phi(\cdot)$ and $U(\cdot)$
are defined in \eqref{ch2:eq:phiPhi} and \eqref{ch2:eq:Mills}.
\end{proposition}

\subsection{Proof of Theorem~\ref{ch2:th:main1}.}
\label{ch2:sec:th:mainmain}

Since the function $v \mapsto \CBS(\kappa, v)$ is a bijection from
$[0,\infty)$ to $[(1-e^\kappa)^+, 1)$,
it admits an inverse function $c \mapsto \VBS(\kappa, c)$, defined by
\begin{equation} \label{ch2:eq:VBS}
	\CBS(\kappa, \VBS(\kappa, c)) = c \,.
\end{equation}
By construction, $\VBS(\kappa, \cdot)$ is a strictly increasing bijection from
$[(1-e^\kappa)^+, 1)$ to $[0,\infty)$.
We will mainly focus on the case $\kappa \ge 0$, for which
$\VBS(\kappa, \cdot): [0,1) \to [0,\infty)$.

\begin{figure}[t]
\centering
\includegraphics[width=.5\columnwidth]{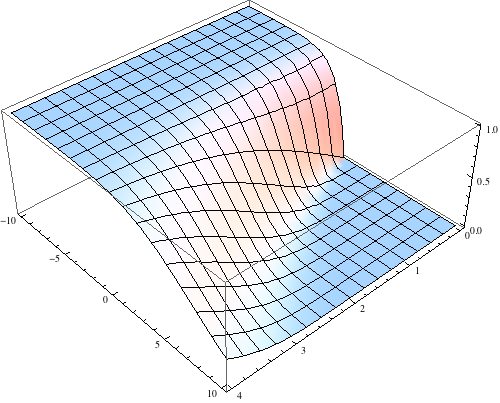}
\caption{A plot of $(\kappa, v) \mapsto \CBS(\kappa, v)$, for
$\kappa \in [-10,10]$ and $v \in [0,4]$.}
\label{ch2:fig:BS}
\end{figure}

\smallskip

Consider an arbitrary model, with a risk-neutral log-price $(X_t)_{t\ge 0}$,
and let $c(\kappa,t)$ be the corresponding
price of a normalized European call option, cf.\ \eqref{ch2:eq:cp}.
Since $z \mapsto (z-e^\kappa)^+$ is a convex function,
one has $c(\kappa,t) \ge (\E[e^{X_t}] - e^\kappa)^+ = (1-e^\kappa)^+$
by Jensen's inequality; since $(z-e^\kappa)^+ < z^+$, one has
$c(\kappa,t) < \E[e^{X_t}] = 1$. 
Consequently,
by \eqref{ch2:eq:VBS}, we have the following relation between the \textit{implied volatility}
$\sigma_\imp(\kappa,t)$ (defined in \S\ref{ch2:sec:setting})
and $\VBS(\kappa,c(\kappa,t))$:
\begin{equation}\label{ch2:eq:IV}
	\sigma_\imp(\kappa,t) := \frac{\VBS(\kappa,c(\kappa,t))}{\sqrt{t}} \,.
\end{equation}

Relation \eqref{ch2:eq:IV} allows to reformulate 
Theorem~\ref{ch2:th:main1} more transparently in terms
of the function $\VBS$. 
Inspired by \eqref{ch2:eq:parity}, we define $p = p(\kappa,c)$ by
\begin{equation} \label{ch2:eq:p}
	p := c - (1-e^\kappa) \,.
\end{equation}
Consider an arbitrary family of values of $(\kappa, c)$,
such that either $\kappa \ge 0$, $c \in (0,1)$ and $c \to 0$,
or alternatively $\kappa \le 0$, $p \in (0,1)$ and $p \to 0$
(with $p$ as in \eqref{ch2:eq:p}).
Then, in light of \eqref{ch2:eq:IV}, we can write the following:

\begin{itemize}
\item If $\kappa$ bounded away from zero $(\liminf |\kappa| > 0)$, 
relation \eqref{ch2:eq:Vas>00} is equivalent to
\begin{equation} \label{ch2:eq:Vas>0}
	\VBS(\kappa,c) \sim
	\begin{cases}
	\rule{0pt}{1.1em}\sqrt{2\, (-\log c + \kappa)} - \sqrt{2\,(- \log c)} & \text{if } \kappa > 0 \,, \\
	\rule{0pt}{1.5em}\sqrt{2\,(-\log p)} - \sqrt{2\, (-\log p + \kappa)} & \text{if } \kappa < 0 \,.
	\end{cases}
\end{equation}

\item If $\kappa$ is bounded away from infinity $(\limsup |\kappa| < \infty)$, 
relations \eqref{ch2:eq:Vas<infty0} and \eqref{ch2:eq:Vas00} are equivalent to
\begin{equation} \label{ch2:eq:Vas<infty}
	\VBS(\kappa,c) \sim 
	\begin{cases}
	\displaystyle\frac{\kappa}{\rule{0pt}{0.95em}D^{-1} (\frac{c}{\kappa})}
	& \text{if } \kappa > 0 \,, \\
	\displaystyle\rule{0pt}{1.7em}
	\sqrt{2\pi} \, c =
	\sqrt{2\pi} \, p & \text{if } \kappa = 0 \,, \\
	\displaystyle\rule{0pt}{1.7em}
	\frac{-\kappa}{\rule{0pt}{0.95em}D^{-1} (\frac{p}{-\kappa})}
	& \text{if } \kappa < 0 \,,
	\end{cases}
\end{equation}
where $D^{-1}(\cdot)$ is the inverse of the function $D(\cdot)$ defined in \eqref{ch2:eq:D},
and satisfies \eqref{ch2:eq:Das}.
\end{itemize}

The proof of Theorem~\ref{ch2:th:main1} is now reduced to
proving relations \eqref{ch2:eq:Vas>0} and \eqref{ch2:eq:Vas<infty}.
We first show that we can assume $\kappa \ge 0$, by a symmetry argument.

\begin{proof}[Deducing the case $\kappa \le 0$ from the case $\kappa \ge 0$]

Recalling \eqref{ch2:eq:CBS-} and \eqref{ch2:eq:VBS},
for all $\kappa\in\R$ and $c \in [(1-e^\kappa)^+,1)$ we have
\begin{equation*}
	\VBS(\kappa,c) = \VBS(-\kappa, 1-e^{-\kappa} + e^{-\kappa}c)
	= \VBS(-\kappa,  e^{-\kappa} p) \,,
\end{equation*}
where $p$ is defined in \eqref{ch2:eq:p}.
As a consequence, in the case $\kappa \le 0$, replacing $\kappa$ by $-\kappa$
and $c$ by $e^{-\kappa} p$ in the first line of \eqref{ch2:eq:Vas>0}, one
obtains the second line of \eqref{ch2:eq:Vas>0}.

Performing the same replacements in the first line of 
\eqref{ch2:eq:Vas<infty} yields
\begin{equation*}
	\VBS(\kappa,c) \sim 
	\frac{-\kappa}{D^{-1} (e^{-\kappa} \frac{p}{-\kappa})} \,,
\end{equation*}
which is slightly different with respect to the third line of \eqref{ch2:eq:Vas<infty}.
However, the discrepancy is only apparent, because we claim that 
$D^{-1} (e^{-\kappa} \frac{p}{-\kappa}) \sim D^{-1} (\frac{p}{-\kappa})$.
This is checked as follows: if $\kappa \to 0$, then
$e^{-\kappa} \frac{p}{-\kappa} \sim \frac{p}{-\kappa}$;
if $\kappa \to \bar\kappa\in (-\infty,0)$,
since $p \to 0$ by assumption, the first relation
in \eqref{ch2:eq:Das} yields
$D^{-1} (e^{-\kappa} \frac{p}{-\kappa}) \sim 
\sqrt{2(-\log (\frac{p}{-\bar\kappa}) + \bar\kappa)} \sim
\sqrt{2(-\log (\frac{p}{-\bar\kappa}))} \sim 
D^{-1} (\frac{p}{-\kappa})$, as required.
(See the lines following \eqref{ch2:eq:claim2} below for more details.)
\end{proof}

\begin{proof}[Proof of \eqref{ch2:eq:Vas>0} for $\kappa \ge 0$]
We fix a family of values of $(\kappa, c)$ with 
$c \to 0$ and $\kappa$ bounded away from zero,
say $\kappa \ge \delta$ for some
fixed $\delta > 0$. Our goal is to prove that relation \eqref{ch2:eq:Vas>0} holds.
If we set $v := \VBS(\kappa,c)$, by definition \eqref{ch2:eq:VBS} we have 
$\CBS(\kappa,v) = c \to 0$.

Let us first show that $d_1  := -\frac{\kappa}{v} + \frac{v}{2} \to -\infty$.
By Proposition~\ref{ch2:th:BS}, $\CBS(\kappa,v) \to 0$ implies
$\min\{d_1, \log v\} \to -\infty$, which means that every subsequence
of values of $(\kappa, c)$ admits a further
sub-subsequence along which either $d_1 \to \infty$ or $v \to 0$.
The key point is that $v \to 0$ implies $d_1 \to -\infty$,
because $d_1 \le -\frac{\delta}{v} + \frac{v}{2}$
(recall that $\kappa \ge \delta$). Thus $d_1 \to -\infty$
along every sub-subsequence, which means that $d_1 \to -\infty$ along
the whole family of values of $(\kappa,c)$.

Since $d_1 \to -\infty$, we can apply relation \eqref{ch2:eq:asd1}.
Taking $\log$ of both sides of that relation, recalling the definition
\eqref{ch2:eq:phiPhi} of $\phi$ and the fact that $\CBS(\kappa,v) = c$, we can write
\begin{equation}\label{ch2:eq:logc}
	\log c \sim -\frac{1}{2} d_1^2 - \log\sqrt{2\pi} + \log\frac{v}{-d_1(-d_1+v)} \,.
\end{equation}
We now show that the last term in the right hand side is $o(d_1^2)$ and can therefore
be neglected. Note that $-d_1 \ge 1$ eventually, because $d_1 \to -\infty$, hence
\begin{equation*}
	\log\frac{v}{-d_1(-d_1+v)} \le  \log\frac{v}{1+v} \le 0 \,.
\end{equation*}
Since $v \mapsto \frac{-d_1+v}{v}$ is decreasing
for $-d_1 > 0$, in case $v \ge -d_1$ one has
\begin{equation*}
	\bigg| \log\frac{v}{-d_1(-d_1+v)}\bigg|
	=  \log\frac{-d_1(-d_1+v)}{v}  \le
	\log (-2d_1) = o(d_1^2) \,.
\end{equation*}
On the other hand, recalling that $d_1 \le -\frac{\delta}{v} + \frac{v}{2}$,
in case $v < -d_1$ one has
$d_1 \le -\frac{\delta}{v} -\frac{d_1}{2}$, which can be rewritten as
$v \ge \frac{2\delta}{-3d_1}$ and together with $v < -d_1$ yields
\begin{equation*}
	\bigg| \log\frac{v}{-d_1(-d_1+v)}\bigg|
	= \log\frac{-d_1(-d_1+v)}{v}  \le \log\frac{-d_1(-d_1-d_1)}{\frac{2\delta}{-3d_1}}  
	= \log\bigg(\frac{3(-d_1)^3}{2\delta}\bigg) = o(d_1^2) \,.
\end{equation*}

In conclusion, \eqref{ch2:eq:logc} yields $\log c \sim -\frac{1}{2} d_1^2$, that is
there exists $\gamma = \gamma(\kappa,c) \to 0$ such that
$(1+\gamma) \log c = -\frac{1}{2} d_1^2$, and since $\log c \le 0$ we can write
\begin{equation*}
	 (1+\gamma) |\log c| = \frac{1}{2} d_1^2 =
	\frac{1}{2} \bigg( \frac{\kappa^2}{v^2} + \frac{v^2}{4}
	- \kappa \bigg)  \,.
\end{equation*}
This is a second degree equation in $v^2$, whose solutions (both positive) are
\begin{equation} \label{ch2:eq:vpm}
	v^2 = 2\kappa \Bigg[ 1 + 2\frac{ (1+\gamma)|\log c|}{\kappa}
	\pm 2\sqrt{\bigg(\frac{ (1+\gamma)|\log c|}{\kappa}\bigg)^2
	+ \frac{ (1+\gamma)|\log c|}{\kappa}} \,\Bigg] \,.
\end{equation}
Since $d_1 \to -\infty$, eventually one has $d_1 < 0$:
since $d_1  = -\frac{\kappa}{v} + \frac{v}{2} = -\frac{1}{2v}(\sqrt{2\kappa}-v)(\sqrt{2\kappa}+v)$,
it follows that $v^2 < 2\kappa$, which selects the ``$-$'' solution in
\eqref{ch2:eq:vpm}. Taking square roots of both sides of \eqref{ch2:eq:vpm} and recalling that
$v = \VBS(\kappa,c)$ yields the equality
\begin{equation} \label{ch2:eq:sqrtt}
	\VBS(\kappa,c) = \sqrt{2(1+\gamma)|\log c| + 2\kappa}
	- \sqrt{2(1+\gamma)|\log c|} \,,
\end{equation}
as one checks squaring both sides of \eqref{ch2:eq:sqrtt}. 

Finally, since $\gamma \to 0$, it is quite intuitive that relation \eqref{ch2:eq:sqrtt} yields
\eqref{ch2:eq:Vas>0}. To prove this fact, we observe that by \eqref{ch2:eq:sqrtt} we can write
\begin{equation} \label{ch2:eq:eheh}
	\frac{\VBS(\kappa,c)}
	{\sqrt{2|\log c| + 2\kappa} - \sqrt{2|\log c|}} =
	f_\gamma\bigg(\frac{\kappa}{|\log c|}\bigg) \,,
\end{equation}
where for fixed $\gamma > -1$ we define
the function $f_\gamma: [0,\infty) \to (0,\infty)$ by
\begin{equation*}
	f_\gamma(x) := \frac{\sqrt{1 + \gamma + x} - \sqrt{1+\gamma}}
	{\sqrt{1 + x} - 1} \quad \text{for } x > 0 \,, \qquad
	f_\gamma(0) := \lim_{x\downarrow 0} f_\gamma(x) = \frac{1}{\sqrt{1+\gamma}} \,.
\end{equation*}
By direct computation, when $\gamma > 0$
(resp.\ $\gamma < 0$) one has $\frac{\dd}{\dd x} f_\gamma(x) > 0$ (resp.\ $< 0$) for all $x > 0$. 
Since $\lim_{x\to+\infty} f_\gamma(x) = 1$, it follows that for every $x \ge 0$ one has
$f_\gamma(0) \le f_\gamma(x) \le 1$ if $\gamma > 0$, while
$1 \le f_\gamma(x) \le f_\gamma(0)$ if $\gamma < 0$; consequently, for any $\gamma$,
\begin{equation*}
	\frac{1}{\sqrt{1+|\gamma|}} \le f_\gamma(x) \le \frac{1}{\sqrt{1-|\gamma|}} \,,
	\qquad \forall x \ge 0 \,,
\end{equation*}
which yields $\lim_{\gamma\to 0}f_\gamma(x) = 1$
\emph{uniformly over $x \ge 0$}.
By \eqref{ch2:eq:eheh}, relation \eqref{ch2:eq:Vas>0} is proved.
\end{proof}

\begin{proof}[Proof of \eqref{ch2:eq:Vas<infty} for $\kappa \ge 0$]
We now fix a family of values of $(\kappa, c)$ with 
$c \to 0$ and $\kappa$ bounded away from infinity, say $0 \le \kappa \le M$ for some
fixed $M \in (0,\infty)$, and we prove relation \eqref{ch2:eq:Vas<infty}.

We set $v := \VBS(\kappa,c)$ so that $\CBS(\kappa,v) = c \to 0$,
cf.\ \eqref{ch2:eq:VBS}. (Note that $v > 0$, because $c > 0$ by assumption.)
Applying Proposition~\ref{ch2:th:BS} we have
$\min\{d_1, \log v\} \to -\infty$, i.e.\ either $d_1 \to -\infty$ or $v \to 0$
along sub-subsequences. However, this time $d_1 \to -\infty$ implies $v \to 0$,
because $d_1 \ge -\frac{M}{v} + \frac{v}{2}$ (recall that $\kappa \le M$),
which means that $v \to 0$ along
the whole given family of values of $(\kappa,c)$.

Since $v \to 0$, relation \eqref{ch2:eq:asv} yields
\begin{equation}\label{ch2:eq:step}
	c \sim - U'(-d_1) \,\phi(d_1) \, v \,.
\end{equation}
Let us focus on $U'(-d_1)$: recalling that
$d_1 = -\frac{\kappa}{v} + \frac{v}{2}$ and $v\to 0$, we first show that
\begin{equation}\label{ch2:eq:firstshow}
	U'(-d_1) \sim U'\bigg(\frac{\kappa}{v}\bigg) \,.
\end{equation}
By a subsequence argument,
we may assume that $\frac{\kappa}{v} \to \rho \in [0,\infty]$,
and we recall that $v\to 0$:
\begin{itemize}
\item  if $\rho < \infty$, $U'(-d_1)$ and $U'(\frac{\kappa}{v})$
converge to $U'(\rho) \ne 0$, hence
$U'(-d_1) / U'(\frac{\kappa}{v}) \to 1$;

\item if $\rho = \infty$, $-d_1$ and
$\frac{\kappa}{v}$ diverge to $\infty$ and \eqref{ch2:eq:Millsas} yields
$U'(-d_1) / U'(\frac{\kappa}{v}) \sim  (\frac{\kappa}{v}) / (-d_1) \to 1$.
\end{itemize}
The proof of \eqref{ch2:eq:firstshow} is completed.
Next we observe that, again by $v\to 0$,
\begin{equation*}
	\phi(-d_1) = \frac{1}{\sqrt{2\pi}} e^{-\frac{1}{2} d_1^2}
	= \frac{1}{\sqrt{2\pi}} e^{-\frac{1}{2} (\frac{\kappa^2}{v^2} + \frac{v^2}{2}
	- \kappa)} \sim 
	e^{\frac{1}{2}\kappa} \,
	\frac{1}{\sqrt{2\pi}} e^{-\frac{1}{2} \frac{\kappa^2}{v^2}}
	= e^{\frac{1}{2}\kappa} \phi\bigg(\frac{\kappa}{v}\bigg) \,.
\end{equation*}
We can thus rewrite \eqref{ch2:eq:step} as
\begin{equation}\label{ch2:eq:quaas}
	c \sim -	U'\bigg(\frac{\kappa}{v}\bigg) \, 
	\phi\bigg(\frac{\kappa}{v}\bigg) \, e^{\frac{1}{2}\kappa} \, v \,.
\end{equation}
If $\kappa = 0$, recalling \eqref{ch2:eq:U12} we obtain
$c \sim \phi(0) v = \frac{1}{\sqrt{2\pi}} v$, which is the second line of \eqref{ch2:eq:Vas<infty}.

Next we assume $\kappa > 0$.
By \eqref{ch2:eq:U12}, \eqref{ch2:eq:Mills} and \eqref{ch2:eq:D}, for all $z > 0$ we can write
\begin{equation*}
	- U'(z) \,\phi(z)
	= -\phi(z) \big(zU(z) - 1\big)
	= \phi(z) - z \Phi(-z)
	= z D(z) \,,
\end{equation*}
hence \eqref{ch2:eq:quaas} can be rewritten as
\begin{equation*}
	c \sim \kappa \, e^{\frac{1}{2}\kappa} \, D\bigg(\frac{\kappa}{v}\bigg) \,,
	\qquad \text{i.e.} \qquad
	(1+\gamma) c = \kappa \, e^{\frac{1}{2}\kappa} \, D\bigg(\frac{\kappa}{v}\bigg) \,,
\end{equation*}
for some $\gamma = \gamma(\kappa,c) \to 0$. Recalling that $v = \VBS(\kappa,c)$,
we have shown that
\begin{equation} \label{ch2:eq:stepql}
	\VBS(\kappa,c) = \frac{\kappa}
	{D^{-1}\Big(\frac{(1+\gamma)c}{\kappa e^{\frac{1}{2}\kappa}}\Big)} \,.
\end{equation}
We now claim that
\begin{equation}\label{ch2:eq:claim2}
	D^{-1}\bigg(\frac{(1+\gamma)c}{\kappa e^{\frac{1}{2}\kappa}}\bigg)
	\sim D^{-1}\bigg(\frac{c}{\kappa}\bigg) \,.
\end{equation}
By a subsequence argument, we may assume that $\frac{c}{\kappa} \to \eta
\in [0,\infty]$ and $\kappa \to \bar\kappa \in [0,M]$.
\begin{itemize}
\item If $\eta \in (0,\infty)$, then $\bar\kappa = 0$
(recall that $c \to 0$) hence $(1+\gamma)c/(\kappa e^{\frac{1}{2}\kappa}) \to \eta$;
then both sides of \eqref{ch2:eq:claim2} converge to
$D^{-1}(\eta) \in (0,\infty)$, hence their ratio converges to $1$.

\item If $\eta = \infty$, then again $\bar\kappa = 0$, hence
$(1+\gamma)c/(\kappa e^{\frac{1}{2}\kappa}) \to \infty$: since
$D^{-1}(y) \sim \frac{1}{\sqrt{2\pi}} y^{-1}$ as $y \to \infty$, cf.\
\eqref{ch2:eq:Das}, it follows immediately that \eqref{ch2:eq:claim2} holds.

\item If $\eta = 0$, then $(1+\gamma)c/(\kappa e^{\frac{1}{2}\kappa}) \to 0$: since
$D^{-1}(y) \sim \sqrt{2|\log y|}$ as $y \to 0$, cf.\ \eqref{ch2:eq:Das},
\begin{equation*}
	D^{-1}\bigg(\frac{(1+\gamma)c}{\kappa e^{\frac{1}{2}\kappa}}\bigg)
	\sim \sqrt{2 \bigg| \bigg(\log \frac{c}{\kappa}\bigg)  +
	\bigg(\log \frac{1+\gamma}{e^{\frac{1}{2}\kappa}}\bigg) \bigg|}
	\sim \sqrt{2 \bigg| \log \frac{c}{\kappa} \bigg|} \,,
\end{equation*}
because $|\log \frac{c}{\kappa}| \to \infty$ while
$|\log [(1+\gamma)/e^{\frac{1}{2}\kappa}]| \to \frac{1}{2} \bar\kappa \in
[0,\frac{M}{2}]$, hence \eqref{ch2:eq:claim2} holds.
\end{itemize}
Having proved \eqref{ch2:eq:claim2}, we can plug it into \eqref{ch2:eq:stepql},
obtaining precisely the first line of \eqref{ch2:eq:Vas<infty}. This
completes the proof of Theorem~\ref{ch2:th:main1}.
\end{proof}

\section{From tail probability to option price}
\label{ch2:sec:probtoprice}

In this section we prove Theorems~\ref{ch2:th:main2b}, \ref{ch2:th:main2bl}
and~\ref{ch2:th:main2a}. We stress that it is enough to prove
the asymptotic relations for the option
prices $c(\kappa,t)$ and $p(-\kappa,t)$, because
the corresponding relations for the implied volatility
$\sigma_\imp(\pm\kappa,t)$ follow immediately applying Theorem~\ref{ch2:th:main1}.

\subsection{Proof of Theorem~\ref{ch2:th:main2b} and~\ref{ch2:th:main2bl}}
\label{ch2:sec:main2bproof}

We prove Theorem~\ref{ch2:th:main2b} and~\ref{ch2:th:main2bl} at the same time.
We recall that the tail probabilities $\F_t(\kappa)$, $F_t(-\kappa)$ are defined
in \eqref{ch2:eq:tail}. Throughout the proof,
we fix a family of values of $(\kappa,t)$ 
with $\kappa > 0$ and $0 < t < T$, for some fixed $T \in (0,\infty)$,
such that Hypothesis~\ref{ch2:ass:rv} is satisfied.

Extracting subsequences, we may distinguish three regimes for $\kappa$:
\begin{itemize}
\item if $\kappa \to \infty$ our goal is to prove \eqref{ch2:eq:ma1c},
resp.\ \eqref{ch2:eq:ma1p};

\item if $\kappa \to \bar\kappa \in (0,\infty)$
our goal is to prove \eqref{ch2:eq:mac}, resp.\ \eqref{ch2:eq:map},
because in this case, plainly, one has
$-\log \F_t(\kappa)/\kappa \to \infty$, resp.\ $-\log F_t(-\kappa)/\kappa \to \infty$,
by \eqref{ch2:eq:Fto0};

\item if $\kappa \to 0$, our goal is to prove \eqref{ch2:eq:ma2c},
resp.\ \eqref{ch2:eq:ma2p}.
\end{itemize}
Of course, each regime has different assumptions,
as in Theorem~\ref{ch2:th:main2b} and~\ref{ch2:th:main2bl}.
 
\medskip
\noindent
\emph{Step 0. Preparation.}
It follows by conditions \eqref{ch2:eq:rv} and \eqref{ch2:eq:cont1} that
\begin{equation}\label{ch2:eq:epsi}
	\forall \epsilon > 0 \quad \exists \rho_\epsilon \in (1,\infty): \qquad
	I_\pm(\rho_\epsilon) < 1+\epsilon \,,
\end{equation}
therefore for every $\epsilon > 0$ one has eventually
\begin{equation}\label{ch2:eq:event}
\begin{split}
	\log \F_t(\rho_\epsilon\kappa) & \ge (1+\epsilon) \log\F_t(\kappa) \,, \quad \ 
	\text{resp.} \\
	\log F_t(-\rho_\epsilon\kappa) & \ge (1+\epsilon) \log F_t(-\kappa) \,,
\end{split}
\end{equation}
where the inequality is ``$\ge$''
instead of ``$\le$'', because both sides are negative quantities.

We stress that $\F_t(\kappa) \to 0$, resp.\ $F_t(-\kappa) \to 0$,
by \eqref{ch2:eq:Fto0}, hence
\begin{equation}\label{ch2:eq:Ftoinfty}
	\log \F_t(\kappa) \to -\infty \,, \quad \ \text{resp.} \quad \
	\log F_t(-\kappa) \to -\infty \,.
\end{equation}
Moreover, we claim that in any of the regimes 
$\kappa \to \infty$, $\kappa \to \bar\kappa \in (0,\infty)$
and $\kappa \to 0$ one has
\begin{equation}\label{ch2:eq:toinf}
	\log \F_t(\kappa) + \kappa \to -\infty \,.
\end{equation}
This follows readily by \eqref{ch2:eq:Ftoinfty} if 
$\kappa \to 0$ or $\kappa \to \bar\kappa \in (0, \infty)$.
If $\kappa \to \infty$ we argue as follows: by Markov's inequality,
for $\eta > 0$
\begin{equation} \label{ch2:eq:Markov}
	\F_t(\kappa) \le \E[e^{(1+\eta)X_t}] e^{-(1+\eta)\kappa} \,, 
\end{equation}
hence
\begin{equation*}
	\log \F_t(\kappa) + \kappa \le - \eta \kappa + \log \E[e^{(1+\eta)X_t}] \,.
\end{equation*}
Since in the regime $\kappa \to \infty$ we assume that the moment
condition \eqref{ch2:eq:moment}
holds for some or every $\eta > 0$, the term $\log \E[e^{(1+\eta)X_t}]$ is bounded from above,
hence eventually
\begin{equation}\label{ch2:eq:bounduseful}
	\log \F_t(\kappa) + \kappa \le -\frac{\eta}{2}\, \kappa  \,,
\end{equation}
which proves relation \eqref{ch2:eq:toinf}.

\smallskip

The rest of the proof is divided in four steps,
in each of which we prove lower and upper bounds
on $c(\kappa,t)$ and $p(-\kappa,t)$, respectively.

\medskip
\noindent
\emph{Step 1. Lower bounds on $c(\kappa,t)$.}
We are going to prove sharp lower bounds on $c(\kappa,t)$,
that will lead to relations \eqref{ch2:eq:ma1c},
\eqref{ch2:eq:mac} and \eqref{ch2:eq:ma2c}.

By \eqref{ch2:eq:cp} and \eqref{ch2:eq:epsi}, 
for every $\epsilon > 0$ we can write
\begin{equation} \label{ch2:eq:stac}
\begin{split}
	c(\kappa,t) & \ge \E[(e^{X_t} - e^\kappa) \ind_{\{X_t > \rho_\epsilon\kappa\}}]
	\ge (e^{\rho_\epsilon\kappa}-e^\kappa) \F_t(\rho_\epsilon\kappa) \,,
\end{split}
\end{equation}
and applying \eqref{ch2:eq:event} we get
\begin{equation} \label{ch2:eq:stacc}
	\log c(\kappa,t) \ge \log \big( e^{\rho_\epsilon\kappa}-e^\kappa \big)
	+ (1+\epsilon) \log \F_t(\kappa) \,.
\end{equation}

If $\kappa \to \infty$, since $\log (e^{\rho_\epsilon\kappa} - e^\kappa) =
\kappa + \log (e^{(\rho_\epsilon-1)\kappa} - 1) \ge \kappa $ eventually, 
we obtain
\begin{equation} \label{ch2:eq:stacc0}
\begin{split}
	\log c(\kappa,t) & \ge \kappa
	+ (1+\epsilon) \log \F_t(\kappa) =
	(1+\epsilon) \big(\log \F_t(\kappa) + \kappa \big) - \epsilon \kappa \\
	& \ge (1+\epsilon +\tfrac{2}{\eta} \epsilon)
	\big( \log \F_t(\kappa) + \kappa \big) \,,
\end{split}
\end{equation}
where in the last inequality we have applied \eqref{ch2:eq:bounduseful}.
It follows that
\begin{equation} \label{ch2:eq:ubc100}
	\limsup \frac{\log c(\kappa,t)}{\log \F_t(\kappa) + \kappa} \le 
	1+\epsilon +\tfrac{2}{\eta}\epsilon \,,
\end{equation}
where the $\limsup$ is taken along the given family of values of $(\kappa,t)$
(note that $\log c(\kappa,t)$ and $\log \F_t(\kappa) + \kappa$ are negative quantities,
cf.\ \eqref{ch2:eq:toinf}, hence the reverse inequality with respect to \eqref{ch2:eq:stacc0}).
Since $\epsilon > 0$ is arbitrary and $\eta > 0$
is fixed, we have shown that
\begin{equation} \label{ch2:eq:ubc1}
	\limsup \frac{\log c(\kappa,t)}{\log \F_t(\kappa) + \kappa} \le 1 \,,
\end{equation}
that is we have obtained a sharp bound for \eqref{ch2:eq:ma1c}.

If $\kappa \to \bar\kappa \in (0,\infty)$, 
since $\log ( e^{\rho_\epsilon\kappa}-e^\kappa) \to \log (e^{\rho_\epsilon\bar\kappa}-
e^{\bar\kappa})$ is bounded while $\log \F_t(\kappa) \to -\infty$,
relation \eqref{ch2:eq:stacc} gives
\begin{equation*}
	\limsup \frac{\log c(\kappa,t)}{\log\F_t(\kappa)}
	\le 1+\epsilon \,.
\end{equation*}
Since $\epsilon > 0$ is arbitrary, we have shown that when $\kappa \to \bar\kappa \in (0,\infty)$
\begin{equation} \label{ch2:eq:ubc}
	\limsup \frac{\log c(\kappa,t)}{\log\F_t(\kappa)}
	\le 1 \,,
\end{equation}
obtaining a sharp bound for \eqref{ch2:eq:mac}.

Finally, if $\kappa \to 0$, since for $\kappa \ge 0$ by convexity
$\log (e^{\rho_\epsilon\kappa} - e^\kappa) =
\kappa + \log (e^{(\rho_\epsilon-1)\kappa} - 1) \ge \kappa + 
\log ((\rho_\epsilon - 1) \kappa)
= \kappa + \log (\rho_\epsilon - 1) + \log \kappa$, relation \eqref{ch2:eq:stacc} yields
\begin{equation*}
	\log \frac{c(\kappa,t)}{\kappa} = \log c(\kappa,t) - \log \kappa
	\ge \log(\rho_\epsilon - 1)
	+(1+\epsilon) \log \F_t(\kappa) \,.
\end{equation*}
Again, since $\log(\rho_\epsilon - 1)$ is constant and
$\log \F_t(\kappa) \to -\infty$, and $\epsilon > 0$ is arbitrary, we get
\begin{equation} \label{ch2:eq:ubc2}
	\limsup \frac{\log \big( c(\kappa,t)/\kappa \big)}{\log \F_t(\kappa)} \le 1 \,,
\end{equation}
proving a sharp bound for \eqref{ch2:eq:ma2c}.

\medskip
\noindent
\emph{Step 2. Lower bounds on $p(-\kappa,t)$.}
We are going to prove sharp lower bounds on $p(-\kappa,t)$,
that will lead to relations \eqref{ch2:eq:ma1p},
\eqref{ch2:eq:map} and \eqref{ch2:eq:ma2p}.

Recalling \eqref{ch2:eq:cp} and \eqref{ch2:eq:epsi}, 
for every $\epsilon > 0$ we can write
\begin{equation} \label{ch2:eq:stap}
\begin{split}
	p(-\kappa,t) & \ge \E[(e^{-\kappa} - e^{X_t}) \ind_{\{X_t \le -\rho_\epsilon\kappa\}}]
	\ge (e^{-\kappa}-e^{-\rho_\epsilon\kappa}) F_t(-\rho_\epsilon\kappa)  \,,
\end{split}
\end{equation}
and applying \eqref{ch2:eq:event} we obtain
\begin{equation}\label{ch2:eq:stapp}
	\log p(-\kappa,t) \ge
	\log \big( e^{-\kappa}-e^{-\rho_\epsilon\kappa} \big)
	+ (1+\epsilon)\log F_t(-\kappa) \,.
\end{equation}

If $\kappa \to \infty$, since
$\log ( e^{-\kappa}-e^{-\rho_\epsilon\kappa} ) =
-\kappa + \log(1-e^{-(\rho_\epsilon - 1)\kappa}) \sim -\kappa$, eventually one has
$\log ( e^{-\kappa}-e^{-\rho_\epsilon\kappa} ) \ge -(1+\epsilon)\kappa$ and we obtain
\begin{equation*}
\begin{split}
	\log p(-\kappa,t) \ge (1+\epsilon)
	\big( \log F_t(-\kappa) - \kappa \big) \,.
\end{split}
\end{equation*}
Since $\epsilon > 0$ is arbitrary, it follows that
\begin{equation} \label{ch2:eq:ubp1}
	\limsup \frac{\log p(-\kappa,t)}{\log F_t(-\kappa) - \kappa}
	\le 1 \,,
\end{equation}
which is a sharp bound for \eqref{ch2:eq:ma1p}.

If $\kappa \to \bar\kappa \in (0,\infty)$, since
$\log ( e^{-\kappa}-e^{-\rho_\epsilon\kappa} ) \to
\log ( e^{-\bar\kappa}-e^{-\rho_\epsilon\bar\kappa} )$ is bounded
while $\log F_t(-\kappa) \to -\infty$, and $\epsilon > 0$ is arbitrary,
relation \eqref{ch2:eq:stapp} gives
\begin{equation} \label{ch2:eq:ubp}
	\limsup \frac{\log p(-\kappa,t)}{\log F_t(-\kappa)}
	\le 1 \,,
\end{equation}
which is a sharp bound for \eqref{ch2:eq:map}.

Finally, if $\kappa \to 0$,
since $e^{-\kappa} - e^{-\rho_\epsilon\kappa} 
= e^{-\rho_\epsilon\kappa} (e^{(\rho_\epsilon - 1)\kappa} -1)
\ge e^{-\rho_\epsilon\kappa}(\rho_\epsilon - 1)\kappa$ by convexity,
since $\kappa \ge 0$, one has eventually
\begin{equation*}
	\log\big( e^{-\kappa} - e^{-\rho_\epsilon\kappa} \big)
	\ge \log \kappa + \log\big(e^{-\rho_\epsilon \kappa}(\rho_\epsilon - 1)\big)
	\ge \log \kappa + \epsilon \log F_t(-\kappa) \,,
\end{equation*}
because $\log\big(e^{-\rho_\epsilon \kappa}(\rho_\epsilon - 1)\big) \to
\log (\rho_\epsilon - 1) > -\infty$ while $\log F_t(-\kappa) \to -\infty$.
Relation \eqref{ch2:eq:stapp} then yields, eventually,
\begin{equation*}
\begin{split}
	\log \frac{p(-\kappa,t)}{\kappa} = \log p(-\kappa,t) - \log \kappa
	& \ge (1+2\epsilon) \log F_t(-\kappa) \,.
\end{split}
\end{equation*}
Since $\epsilon > 0$ is arbitrary, we have shown that
\begin{equation} \label{ch2:eq:ubp2}
	\limsup \frac{\log \big(p(-\kappa,t) / \kappa\big)}{\log F_t(-\kappa)}
	\le 1 \,,
\end{equation}
obtaining a sharp bound for \eqref{ch2:eq:ma2p}.

\medskip
\noindent
\emph{Step 3. Upper bounds on $c(\kappa,t)$.}
We are going to prove sharp upper bounds on $c(\kappa,t)$,
that will complete the proof of relations \eqref{ch2:eq:ma1c},
\eqref{ch2:eq:mac} and \eqref{ch2:eq:ma2c}.
We first consider the case when \emph{the moment assumptions
\eqref{ch2:eq:moment} and \eqref{ch2:eq:moment0p}
hold for every $\eta > 0$}.

\smallskip

Let us look at the regimes
$\kappa \to \infty$ and $\kappa \to \bar\kappa \in (0,\infty)$
(i.e.\, $\kappa$ is bounded away from zero),
assuming that condition \eqref{ch2:eq:moment} holds \emph{for every $\eta > 0$}.
By H\"older's inequality,
\begin{equation} \label{ch2:eq:deco20}
	c(\kappa,t) 
	= \E[(e^{X_t}-e^\kappa) \ind_{\{X_t >\kappa\}}]
	\le \E[e^{X_t} \ind_{\{X_t >\kappa\}}] 
	\le	\E[e^{(1+\eta)X_t}]^{\frac{1}{1+\eta}} \, \F_t(\kappa)^{\frac{\eta}{1+\eta}} \,.
\end{equation}
Let us fix $\epsilon > 0$ and choose $\eta = \eta_\epsilon$
large enough, so that $\frac{\eta}{1+\eta} > 1-\epsilon$. By assumption \eqref{ch2:eq:moment},
for some $C \in (0,\infty)$ one has
\begin{equation*}
	\E[e^{(1+\eta)X_t}]^{\frac{1}{1+\eta}} \le C \,,
\end{equation*}
hence eventually, recalling that $\log \F_t(\kappa) \to -\infty$, by \eqref{ch2:eq:Ftoinfty},
\begin{equation}\label{ch2:eq:ah2b}
\begin{split}
	\log c(\kappa,t) & \le
	\log C + (1-\epsilon) \log \F_t(\kappa)
	\le (1-2\epsilon) \log \F_t(\kappa) \,.
\end{split}
\end{equation}
Since $\epsilon > 0$ is arbitrary, this shows that
\begin{equation} \label{ch2:eq:lbc10}
	\liminf \frac{\log c(\kappa,t)}{\log \F_t(\kappa)}
	\ge 1 \,.
\end{equation}
which together with \eqref{ch2:eq:ubc} completes the proof of \eqref{ch2:eq:mac},
if $\kappa \to \bar\kappa \in (0,\infty)$.
If $\kappa \to \infty$ and condition \eqref{ch2:eq:moment} holds for every $\eta > 0$,
then $\log \F_t(\kappa)/\kappa \to -\infty$ by \eqref{ch2:eq:Markov}, which yields
$\log \F_t(\kappa) \sim \log \F_t(\kappa) + \kappa$,
hence \eqref{ch2:eq:lbc10} together with \eqref{ch2:eq:ubc1} completes the proof of
\eqref{ch2:eq:ma1c}.

We then consider the regime $\kappa \to 0$,
assuming that condition \eqref{ch2:eq:moment0p} holds \emph{for every $\eta > 0$}.
We modify \eqref{ch2:eq:deco20} as follows:
since $(e^{X_t}-e^\kappa) \le (e^{X_t}- 1) \le |e^{X_t}-1|$,
\begin{equation} \label{ch2:eq:ah20}
	c(\kappa,t) \le \E[|e^{X_t}-1| \ind_{\{X_t >\kappa \}}] \le \kappa \,
	\E\bigg[ \bigg|\frac{e^{X_t}-1}{\kappa}\bigg|^{1+\eta}
	\bigg]^{\frac{1}{1+\eta}} \, \F_t(\kappa)^{\frac{\eta}{1+\eta}} \,.
\end{equation}
Let us fix $\epsilon > 0$ and choose $\eta = \eta_\epsilon$
large enough, so that $\frac{\eta}{1+\eta} > 1-\epsilon$. By assumption 
\eqref{ch2:eq:moment0p},
for some $C \in (0,\infty)$ one has
\begin{equation} \label{ch2:eq:ollaa}
	\E\bigg[ \bigg|\frac{e^{X_t}-1}{\kappa}\bigg|^{1+\eta}
	\bigg]^{\frac{1}{1+\eta}} \le C \,,
\end{equation}
hence relation \eqref{ch2:eq:ah20} yields eventually
\begin{equation}\label{ch2:eq:ah20b}
	\log \frac{c(\kappa,t)}{\kappa} \le
	\log C + (1-\epsilon) \log \F_t(\kappa)
	\le (1-2\epsilon) \log \F_t(\kappa) \,.
\end{equation}
Since $\epsilon > 0$ is arbitrary, we have proved that
\begin{equation} \label{ch2:eq:lbc2}
	\liminf \frac{\log \big( c(\kappa,t)/\kappa \big)}{\log \F_t(\kappa)} \ge 1 \,,
\end{equation}
which together with \eqref{ch2:eq:ubc2} completes the proof of \eqref{ch2:eq:ma2c}.

\smallskip

It remains to consider the case when the moment assumptions
\eqref{ch2:eq:moment} and \eqref{ch2:eq:moment0p}
holds \emph{for some $\eta > 0$}, but in addition
conditions \eqref{ch2:eq:Iplus} (if $\kappa \to \infty$ or
$\kappa \to \bar\kappa \in (0,\infty)$)
or \eqref{ch2:eq:I+infty} (if $\kappa \to 0$) holds.
We start with considerations that are valid in any regime of $\kappa$. 

Defining the constant
\begin{equation} \label{ch2:eq:Amax0}
	A := \limsup \bigg\{ \frac{-\kappa}{\log\F_t(\kappa) + \kappa} \bigg\} + 1\,,
\end{equation}
where the $\limsup$ is taken along the given family of values of $(\kappa,t)$,
we claim that $A < \infty$. This follows by \eqref{ch2:eq:toinf}
if $\kappa \to 0$ or if $\kappa \to \bar\kappa \in (0,\infty)$
(in which case, plainly, $A=1$), while if
$\kappa \to +\infty$ it suffices to apply \eqref{ch2:eq:bounduseful}
to get $A \le 2/\eta + 1$. It follows by \eqref{ch2:eq:Amax0} that eventually
\begin{equation} \label{ch2:eq:Amax}
	\kappa \le -A (\log\F_t(\kappa) + \kappa) \,.
\end{equation}

Next we show that, for all fixed $\epsilon > 0$ and $1 < M < \infty$, eventually one has
\begin{equation}\label{ch2:eq:claiM}
	\log \bigg(\sup_{y \in [1,M]} e^{\kappa y} \, \F_t(\kappa y) \bigg)
	\le (1-\epsilon) \big( \log \F_t(\kappa) + \kappa \big) \,,
\end{equation}
which means that the $\sup$ is approximately attained for $y = 1$.
This is easy if $\kappa \to 0$ or if $\kappa \to \bar\kappa \in (0,\infty)$:
in fact, since $\kappa \to \F_t(\kappa)$ is non-increasing, we can write
\begin{equation*}
\begin{split}
	\log \bigg(\sup_{y \in [1,M]} e^{\kappa y} \, \F_t(\kappa y) \bigg)
	& \le \log \big( e^{\kappa M} \F_t(\kappa) \big)
	= \kappa M + \log \F_t(\kappa) \\
	& = \big( \log \F_t(\kappa) + \kappa \big) 
	\,+\, (M-1) \kappa \,,
\end{split}
\end{equation*}
and since $\log \F_t(\kappa) + \kappa \to -\infty$ by \eqref{ch2:eq:toinf},
while $(M-1) \kappa$ is bounded, \eqref{ch2:eq:claiM} follows.

To prove \eqref{ch2:eq:claiM} in the regime $\kappa\to\infty$,
we are going to exploit the assumption \eqref{ch2:eq:Iplus}.
First we fix $\delta > 0$, to be defined later, and set
$\bar n := \lceil \frac{M-1}{\delta} \rceil$ and $a_n := 1 + n \delta$ for $n=0, \ldots, \bar n$,
so that $[1,M] \subseteq \bigcup_{n=1}^{\bar n} [a_{n-1}, a_n]$. For all $y \in [a_{n-1}, a_n]$
one has, by \eqref{ch2:eq:rv},
\begin{equation*}
	\log \F_t(\kappa y) \le \log \F_t(\kappa a_{n-1}) 
	\sim I_+(a_{n-1}) \log \F_t(\kappa) \le a_{n-1} \log \F_t(\kappa) \,,
\end{equation*}
having used that $I_+(\rho) \ge \rho$, by \eqref{ch2:eq:Iplus}, hence eventually
\begin{equation*}
	\log \F_t(\kappa y) \le (1-\delta) a_{n-1} \log \F_t(\kappa) \,, \qquad
	\forall y \in [a_{n-1}, a_n] \,.
\end{equation*}
Recalling that $a_n = a_{n-1} + \delta$, we can write
$a_n 
\le (1-\delta) a_{n-1} + \delta (1+M)$, because $a_{n-1} \le M$ by construction,
and since $e^{\kappa y} \le e^{\kappa a_n}$ for $y \in [a_{n-1}, a_n]$,
it follows that
\begin{equation*}
\begin{split}
	\log \bigg(\sup_{y \in [1,M]} e^{\kappa y} \, \F_t(\kappa y) \bigg)
	& \le \max_{n=1, \ldots, \bar n} \big( a_n \kappa + 
	(1-\delta) a_{n-1} \log \F_t(\kappa)
	\big) \\
	& = \max_{n=1, \ldots, \bar n} \big((1-\delta) a_{n-1} \big( \log \F_t(\kappa)
	+ \kappa \big) + \delta(1+M) \kappa \big) \,.
\end{split}
\end{equation*}
Plainly, the $\max$ is attained for $n=1$, for which $a_{n-1} = a_0 = 1$.
Recalling \eqref{ch2:eq:Amax}, we get
\begin{equation*}
	\log \bigg(\sup_{y \in [1,M]} e^{\kappa y} \, \F_t(\kappa y) \bigg)
	\le (1-\delta(1+A+AM)) \big( \log \F_t(\kappa)
	+ \kappa \big)  \,.
\end{equation*}
Choosing $\delta := \epsilon / (1+A+AM)$, the claim \eqref{ch2:eq:claiM} is proved.

\smallskip

We are ready to give sharp upper bounds on $c(\kappa,t)$,
refining \eqref{ch2:eq:deco20}. For fixed $M\in (0,\infty)$, we write
\begin{equation} \label{ch2:eq:deco2}
	c(\kappa,t) 
	= \E[(e^{X_t}-e^\kappa) \ind_{\{\kappa < X_t \le \kappa M\}}]
	+ \E[(e^{X_t}-e^\kappa) \ind_{\{X_t >\kappa M\}}]  \,,
\end{equation}
and we estimate the first term as follows: by Fubini-Tonelli's theorem
and \eqref{ch2:eq:claiM},
\begin{equation} \label{ch2:eq:cnew0}
\begin{split}
	\E[(e^{X_t} - e^\kappa) \ind_{\{\kappa < X_t \le \kappa M\}}] 
	& = \E\bigg[ \bigg( \int_\kappa^{\infty} e^x \, \ind_{\{x < X_t\}} \, \dd x
	\bigg) \ind_{\{\kappa < X_t \le \kappa M\}} \bigg] \\
	& = \int_\kappa^{\kappa M} e^x \, \P(x < X_t \le \kappa M)\, \dd x  
	\le \int_\kappa^{\kappa M} e^x \, \F_t(x) \, \dd x  \\
	& = \kappa \int_1^M e^{\kappa y} \, \F_t(\kappa y) \, \dd y 
	\le  \kappa \, (M-1) \,
	e^{(1-\epsilon)( \log \F_t(\kappa) + \kappa )} \,.
\end{split}
\end{equation}
To estimate the second term in \eqref{ch2:eq:deco2},
we start with the cases $\kappa \to \infty$ and $\kappa \to \bar\kappa \in (0,\infty)$,
where we assume that \eqref{ch2:eq:moment} holds for some $\eta > 0$,
as well as \eqref{ch2:eq:I+infty},
hence we can fix $M > 1$ such that $I_+(M) > \frac{1+\eta}{\eta}$.
Bounding $(e^{X_t}-e^\kappa) \le e^{X_t}$,
H\"older's inequality yields
\begin{equation*}
	\E[(e^{X_t}-e^\kappa) \ind_{\{X_t >\kappa M\}}]
	\le \E[e^{(1+\eta)X_t}]^{\frac{1}{1+\eta}}
	\, \F_t(\kappa M)^{\frac{\eta}{1+\eta}}
	= C \, \F_t(\kappa M)^{\frac{\eta}{1+\eta}} \,,
\end{equation*}
where $C \in (0,\infty)$ is an absolute constant, by \eqref{ch2:eq:moment}.
Applying relation \eqref{ch2:eq:rv}
together with $I_+(M) > \frac{1+\eta}{\eta}$ we obtain
\begin{equation} \label{ch2:eq:eventu-1}
	\frac{\eta}{1+\eta} \log \F_t(\kappa M)
	\sim \frac{\eta}{1+\eta} I_+(M) \log \F_t(\kappa)
	\le \log \F_t(\kappa) \,,
\end{equation}
hence eventually
\begin{equation} \label{ch2:eq:eventu}
	\log \E[(e^{X_t}-e^\kappa) \ind_{\{X_t >\kappa M\}}]
	\le (1-\epsilon) \log \F_t(\kappa)
	\le (1-\epsilon) \big(\log \F_t(\kappa) + \kappa \big) \,.
\end{equation}
Recalling \eqref{ch2:eq:bounduseful} and \eqref{ch2:eq:toinf},
eventually $\kappa(M-1) \le 
e^{-\epsilon (\log \F_t(\kappa) + \kappa)}$, hence by \eqref{ch2:eq:cnew0}
\begin{equation} \label{ch2:eq:eventu2}
	\log \E[(e^{X_t} - e^\kappa) \ind_{\{\kappa < X_t \le \kappa M\}}] 
	\le (1-2\epsilon) \big(\log \F_t(\kappa) + \kappa \big) \,.
\end{equation}
Looking back at \eqref{ch2:eq:deco2}, since
\begin{equation}\label{ch2:eq:logab}
	\log (a+b) \le \log 2 + \max\{\log a,\log b\} \,,
	\qquad \forall a,b > 0 \,,
\end{equation}
by \eqref{ch2:eq:eventu}, \eqref{ch2:eq:eventu2}
and again \eqref{ch2:eq:toinf} one has eventually
\begin{equation*}
	\log c(\kappa,t) \le \log 2 + 
	(1-2\epsilon) \big(\log \F_t(\kappa) + \kappa \big)
	\le (1-3\epsilon) \big(\log \F_t(\kappa) + \kappa \big) \,.
\end{equation*}
Since $\epsilon > 0$ is arbitrary, this shows that
\begin{equation} \label{ch2:eq:lbc1}
	\liminf \frac{\log c(\kappa,t)}{\log \F_t(\kappa) + \kappa}
	\ge 1 \,,
\end{equation}
which together with \eqref{ch2:eq:ubc1} completes the proof of
\eqref{ch2:eq:ma1c}, if $\kappa \to \infty$.
Since $\log \F_t(\kappa) + \kappa \sim \log \F_t(\kappa)$
if $\kappa \to \bar\kappa \in (0,\infty)$, by \eqref{ch2:eq:Ftoinfty},
we can rewrite \eqref{ch2:eq:lbc1} in this case as
\begin{equation} \label{ch2:eq:lbc}
	\liminf \frac{\log c(\kappa,t)}{\log \F_t(\kappa)}
	\ge 1 \,,
\end{equation}
which together with \eqref{ch2:eq:ubc} completes the proof of \eqref{ch2:eq:mac}.

It remains to consider the case when $\kappa \to 0$,
where we assume that relation \eqref{ch2:eq:moment0p} holds
for some $\eta \in (0,\infty)$, together with \eqref{ch2:eq:I+infty}.
As before, we fix $M > 1$ such that $I_+(M) > \frac{1+\eta}{\eta}$.
Since
\begin{equation} \label{ch2:eq:compan0}
	\E\bigg[ \bigg(\frac{e^{X_t}-e^\kappa}{\kappa} 
	\bigg)^{1+\eta}  \ind_{\{X_t > \kappa\}} \bigg] \le 
	\E\bigg[ \bigg|\frac{e^{X_t}-1}{\kappa} 
	\bigg|^{1+\eta}\bigg] \le C \,,
\end{equation}
for some absolute constant $C \in (0,\infty)$, by \eqref{ch2:eq:moment0p},
the second term in \eqref{ch2:eq:deco2} is bounded by
\begin{equation} \label{ch2:eq:compan1}
	\E[(e^{X_t}-e^\kappa) \ind_{\{X_t >\kappa M\}}]
	\le \kappa \E\bigg[ \bigg|\frac{e^{X_t}-e^\kappa}{\kappa} 
	\bigg|^{1+\eta} \bigg]^{\frac{1}{1+\eta}}
	\, \F_t(\kappa M)^{\frac{\eta}{1+\eta}}
	\le \kappa \, C \, \F_t(\kappa M)^{\frac{\eta}{1+\eta}} \,.
\end{equation}
In complete analogy with \eqref{ch2:eq:eventu-1}-\eqref{ch2:eq:eventu}, we obtain that eventually
\begin{equation} \label{ch2:eq:compan2}
	\log \frac{\E[(e^{X_t}-e^\kappa) \ind_{\{X_t >\kappa M\}}]}{\kappa} \le (1-\epsilon)
	\log \F_t(\kappa) \,.
\end{equation}
By \eqref{ch2:eq:toinf}, eventually $(M-1) \le 
e^{-\epsilon (\log \F_t(\kappa) + \kappa)}$, hence by \eqref{ch2:eq:cnew0}
\begin{equation}\label{ch2:eq:cnew}
	\log \frac{\E[(e^{X_t} - e^\kappa) \ind_{\{\kappa < X_t \le \kappa M\}}]}{\kappa}
	\le (1-2\epsilon)( \log \F_t(\kappa) + \kappa ) \,.
\end{equation}
Recalling \eqref{ch2:eq:deco2} and \eqref{ch2:eq:logab}, we can finally write
\begin{equation*}
	\log \frac{c(\kappa,t)}{\kappa} \le \log 2 + 
	(1-2\epsilon) \big(\log \F_t(\kappa) + \kappa \big)
	\le (1-3\epsilon) \log \F_t(\kappa) \,,
\end{equation*}
because $\kappa \to 0$ and $\log\F_t(\kappa) \to - \infty$.
Since $\epsilon > 0$ is arbitrary, we have proved that
\begin{equation} \label{ch2:eq:lbc2bis}
	\liminf \frac{\log \big( c(\kappa,t)/\kappa \big)}{\log \F_t(\kappa)} \ge 1 \,,
\end{equation}
which together with \eqref{ch2:eq:ubc2} completes the proof of \eqref{ch2:eq:ma2c}.

\medskip
\noindent
\emph{Step 4. Upper bounds on $p(-\kappa,t)$.}
We are going to prove sharp upper bounds on $p(-\kappa,t)$,
that will complete the proof of relations \eqref{ch2:eq:ma1p},
\eqref{ch2:eq:map} and \eqref{ch2:eq:ma2p}.

By \eqref{ch2:eq:cp} we can write
\begin{equation*}
	p(-\kappa,t) = \E[(e^{-\kappa} - e^{X_t}) \ind_{\{X_t \le -\kappa\}}]
	\le e^{-\kappa} \, F_t(-\kappa) \,,
\end{equation*}
therefore
\begin{equation} \label{ch2:eq:lbp1}
	\frac{\log p(-\kappa,t)}{\log F_t(-\kappa) - \kappa} \ge 1 \,,
\end{equation}
which together with \eqref{ch2:eq:ubp1} 
completes the proof of \eqref{ch2:eq:ma2p},
if $\kappa \to \infty$. On the other hand,
if $\kappa \to \bar\kappa \in (0,\infty)$, since
relation \eqref{ch2:eq:lbp1} implies (recall that $\kappa\ge 0$)
\begin{equation} \label{ch2:eq:lbp}
	\frac{\log p(-\kappa,t)}{\log F_t(-\kappa)} \ge 1 \,,
\end{equation}
in view of \eqref{ch2:eq:ubp},
the proof of \eqref{ch2:eq:map} is completed.

It remains to consider the case $\kappa \to 0$.
If relation \eqref{ch2:eq:moment0p} holds \emph{for every $\eta \in (0,\infty)$},
we argue in complete analogy with 
\eqref{ch2:eq:ah20}-\eqref{ch2:eq:ollaa}-\eqref{ch2:eq:ah20b}, getting
\begin{equation} \label{ch2:eq:allo}
	\liminf \frac{\log \big( p(-\kappa,t) / \kappa \big)}{\log F_t(-\kappa) } \ge 1 \,,
\end{equation}
which together with \eqref{ch2:eq:ubp2} completes the proof of \eqref{ch2:eq:ma2p}.
If, on the other hand, relation \eqref{ch2:eq:moment0p} holds only
\emph{for some $\eta \in (0,\infty)$}, we also assume that condition
\eqref{ch2:eq:I-infty} holds, hence we can fix $M > 1$ such that $I_-(M) > \frac{1+\eta}{\eta}$.
Let us write
\begin{equation} \label{ch2:eq:deco2bis}
	p(-\kappa,t) = \E[(e^{-\kappa} - e^{X_t}) \ind_{\{-\kappa M < X_t \le -\kappa\}}]
	+ \E[(e^{-\kappa} - e^{X_t}) \ind_{\{X_t \le -\kappa M\}}]  \,.
\end{equation}
In analogy with \eqref{ch2:eq:cnew0},
for every fixed $\epsilon > 0$,
the first term in the right hand side can be estimated as follows
(note that $y \mapsto F_t(-\kappa y)$ is decreasing):
\begin{equation*}
\begin{split}
	\E[(e^{-\kappa} - e^{X_t}) \ind_{\{-\kappa M < X_t \le -\kappa\}}]
	& \le \int_{-\kappa M}^{-\kappa} e^x \, F_t(x) \, \dd x
	= \kappa \int_{1}^{M} e^{-\kappa y} \, F_t(-\kappa y) \, \dd y \\
	& \le \kappa (M-1) \, F_t(-\kappa)
	\le \kappa \, e^{(1-\epsilon) \log F_t(-\kappa)} \,.
\end{split}
\end{equation*}
The second term in \eqref{ch2:eq:deco2bis} is estimated in complete analogy
with \eqref{ch2:eq:compan0}-\eqref{ch2:eq:compan1}-\eqref{ch2:eq:compan2}, yielding
\begin{equation*}
	\log \frac{\E[(e^{-\kappa} - e^{X_t}) \ind_{\{X_t \le -\kappa M\}}]}{\kappa}
	\le (1-\epsilon) \log F_t(-\kappa) \,.
\end{equation*}
Recalling \eqref{ch2:eq:logab}, we obtain from 
\eqref{ch2:eq:deco2bis}
\begin{equation*}
	\log \frac{p(-\kappa,t)}{\kappa} \le \log 2 + (1-\epsilon) \log F_t(-\kappa)
	\le (1-2\epsilon) \log F_t(-\kappa) \,,
\end{equation*}
and since $\epsilon > 0$ is arbitrary we have proved that
relation \eqref{ch2:eq:allo} still holds,
which together with \eqref{ch2:eq:ubp2} completes the proof of \eqref{ch2:eq:ma2p},
and of the whole Theorem~\ref{ch2:th:main2b}.\qed

\subsection{Proof of Theorem~\ref{ch2:th:main2a}}
\label{ch2:sec:proofth:main2a}

By Skorokhod's representation theorem, 
we can build a coupling of the random variables $(X_t)_{t\ge 0}$ and $Y$
such that relation \eqref{ch2:eq:scaling} holds a.s.. 
Since the function $z \mapsto z^+$ is continuous,
recalling that $\gamma_t \to 0$, for $\kappa \sim a \gamma_t$ we have a.s.
\begin{equation} \label{ch2:eq:both}
	\frac{(e^{X_t}-e^\kappa)^+}{\gamma_t} =
	\bigg( \frac{e^{Y \gamma_t(1+o(1))}-1}{\gamma_t} - 
	\frac{e^{a \gamma_t(1+o(1))}-1}{\gamma_t} \bigg)^+
	\xrightarrow[t\downarrow 0]{\,a.s.\,} (Y - a)^+ \,,
\end{equation}
and analogously for $\kappa \sim -a \gamma_t$
\begin{equation} \label{ch2:eq:both2}
	\frac{(e^{\kappa} - e^{X_t})^+}{\gamma_t} 
	\xrightarrow[t\downarrow 0]{\,a.s.\,} (-a-Y)^+ = (Y+a)^- \,.
\end{equation}
Taking the expectation of both sides of these relations, one would obtain
precisely \eqref{ch2:eq:astypical}.
To justify the interchanging of limit and expectation, we observe that
the left hand sides of \eqref{ch2:eq:both} and \eqref{ch2:eq:both2} are uniformly integrable,
being bounded in $L^{1+\eta}$. In fact
\begin{equation*}
	\frac{|e^{X_t} - e^\kappa|}{\gamma_t} \le
	\frac{|e^{X_t} - 1|}{\gamma_t} + \frac{|e^\kappa - 1|}{\gamma_t} \,,
\end{equation*}
and the second term in the right hand side is uniformly bounded
(recall that $\kappa \sim a \gamma_t$ by assumption), while
the first term is bounded in $L^{1+\eta}$, by \eqref{ch2:eq:assX}.
\qed

\appendix

\section{Miscellanea}
\label{ch2:sec:app}

\subsection{About conditions \eqref{ch2:eq:assfamily} and \eqref{ch2:eq:cpvanish}}
\label{ch2:sec:cpexplained}

Recall from \S\ref{ch2:sec:setting}
that $(X_t)_{t\ge 0}$ denotes the risk-neutral log-price,
and  assume that $X_t \to X_0 := 0$ in distribution as $t \to 0$
(which is automatically satisfied if $X$ has right-continuous paths).
For an arbitrary family
of values of $(\kappa,t)$, with $t > 0$ and $\kappa \ge 0$,
we show that condition
\eqref{ch2:eq:assfamily} implies \eqref{ch2:eq:cpvanish}.

Assume first that $t \to 0$ (with no assumption on $\kappa$).
Since $\kappa \ge 0$, one has
$(e^{X_t}-e^\kappa)^+ \to (1- e^\kappa)^+ = 0$ in distribution,
hence $c(\kappa,t) \to 0$ by \eqref{ch2:eq:cp} and Fatou's lemma.
With analogous arguments, one has
$p(-\kappa,t) \to 0$, hence \eqref{ch2:eq:cpvanish} is satisfied.

Next we assume that $\kappa \to \infty$ and $t$ is bounded,
say $t \in (0,T]$ for some fixed $T > 0$.
Since $z \mapsto (z - c)^+$ is a convex function and $(e^{X_t})_{t\ge 0}$
is a martingale, the process $((e^{X_t} - e^\kappa)^+)_{t\ge 0}$ is a submartingale
and by \eqref{ch2:eq:cp} we can write
\begin{equation*}
	0 \le c(\kappa,t) \le \E[(e^{X_T} - e^\kappa)^+]
	= \E[(e^{X_T} - e^\kappa)\ind_{\{X_T > \kappa\}}]
	\le \E[e^{X_T}\ind_{\{X_T > \kappa\}}] \,.
\end{equation*}
It follows that, if $\kappa \to +\infty$, then $c(\kappa,t) \to 0$. 
With analogous arguments, one shows that
$p(-\kappa,t) \to 0$, hence condition \eqref{ch2:eq:cpvanish} holds.

\subsection{About Remark~\ref{rem:equivalent}}
\label{app:diffusions}

Let $(S_t)_{t\ge 0}$ be a positive process which solves \eqref{eq:sde}.
By Ito's formula, the process $X_t := \log S_t$ solves
\begin{equation}\label{eq:sde2}
	\begin{cases}
	\dd X_t = \sqrt{V_t} \, \dd W_t - \frac{1}{2} V_t \, \dd t \\
	X_0 = 0
	\end{cases} \,.
\end{equation}
Assuming $V_t \to \sigma_0^2$ a.s.\ as $t \to 0$,
we want to show that 
\begin{equation}\label{eq:godis}
	\frac{X_t}{\sqrt{t}} \xrightarrow[t\to 0]{\, d \,} Y \sim N(0,\sigma_0^2) \,.
\end{equation}

Let us define
\begin{equation*}
	J_t := \frac{1}{2 \sqrt{t}} \int_0^t V_s \, \dd s  \,, \qquad
	I_t := \frac{X_t}{\sqrt{t}}  + J_t - \sigma_0 \frac{W_t}{\sqrt{t}}
	= \int_0^t \frac{\sqrt{V_s} - \sigma_0}{\sqrt{t}} \, \dd W_s
	\,.
\end{equation*}
By $V_t \to \sigma_0^2$ a.s.\ it follows that
$J_t \sim \frac{\sqrt{t}}{2} \sigma_0 \to 0$ a.s., by the fundamental theorem
of calculus. Moreover, since $\sqrt{V_t} \to \sigma_0$ a.s.,
\begin{equation}\label{eq:quadra}
	\langle I \rangle_t := \int_0^t \frac{|\sqrt{V_s} - \sigma_0|^2}{t} \, \dd s
	\le \sup_{0 \le s \le t} |\sqrt{V_s} - \sigma_0|^2
	\xrightarrow[\,t\to 0\,]{a.s.} 0 \,.
\end{equation}
We now use the inequality $\P(|I_t| > \epsilon) \le \frac{\delta}{\epsilon^2}
+ \P(\langle I \rangle_t > \delta)$,
cf.\ \cite[Problem 5.25]{cf:KS88}. Sending first $t \to 0$ for fixed $\delta > 0$,
and then $\delta \to 0$, we see by \eqref{eq:quadra} that
$I_t \to 0$ in probability as $t \to 0$.
Since $\sigma_0 W_t /\sqrt{t} \to Y \sim N(0,\sigma_0^2)$
in distribution as $t \to 0$,\footnote{In fact, 
the distribution of $\sigma_0 W_t /\sqrt{t}$ is $N(0,\sigma_0^2)$ \emph{for all $t > 0$}.}
by Slutsky's theorem
\begin{equation*}
	\frac{X_t}{\sqrt{t}} =
	I_t - J_t
	+ \sigma_0 \frac{W_t}{\sqrt{t}} \xrightarrow[t\to 0]{d} 0 + Y = Y \sim N(0,\sigma_0^2) \,,
\end{equation*}
hence relation \eqref{eq:godis} holds.

\smallskip

Next we show that, plugging $Y \sim N(0,\sigma_0^2)$ into \eqref{ch2:eq:sigma0},
we obtain $C_\pm(a) = \sigma_0$ for any $a \ge 0$. Since $Y$ has a symmetric
law, it suffices to focus on $C_+(a)$. Then
\begin{equation} \label{eq:DD-1}
 \begin{split}
\E[(Y-a)^+] & = \sigma_0 \,\E\left[\left(N(0,1)-\frac{a}{\sigma_0}\right)^+\right]
= \sigma_0 \bigg[\int_{\frac{a}{\sigma_0}}^{\infty}
x \, \frac{e^{-\frac{x^2}{2}}}{\sqrt{2\pi}} \, \dd x
\,-\, \frac{a}{\sigma_0} \int_{\frac{a}{\sigma_0}}^{\infty} 
\frac{e^{-\frac{x^2}{2}}}{\sqrt{2\pi}} \, \dd x \bigg]\\
& = \sigma_0\Bigg(\phi\bigg(\frac{a}{\sigma_0}\bigg) \,-\,
\frac{a}{\sigma_0}\Phi\left(-\frac{a}{\sigma_0}\right)\Bigg) 
= a \, D\bigg(\frac{a}{\sigma_0}\bigg)\,,
 \end{split}
\end{equation}
where we used the density $\phi$ and distribution function $\Phi$
of a $N(0,1)$ random variable, cf.\ \eqref{ch2:eq:phiPhi},
and the definition \eqref{ch2:eq:D} of $D$.
Looking back at \eqref{ch2:eq:sigma0}, we obtain $C_+(a) = \sigma_0$.\qed

\subsection{Proof of Lemma~\ref{th:lemma}}
\label{sec:lemma}

We start with some estimates.
It follows by \eqref{ch2:eq:MertonGeneratingFunction} that
\[
 X_t \overset{d}{=} \sigma W_t+\mu t+\sum_{i=1}^{N_t}Y_i
\]
with $Y_i \sim N(\alpha, \delta^2)$ and $N_t \sim Pois(\lambda t)$
(and we agree that the sum equals $0$ in case $N_t = 0$).
By Chernoff's 
bound\footnote{Just apply Markov's inequality 
$\P(N_t > M) \le e^{-M \alpha} \E[e^{\alpha N_t}] =
e^{-M \alpha + \lambda t(e^\alpha - 1)}$ and optimize over $\alpha \ge 0$.} 
$\P(N_t > M) \le (\frac{e\lambda t}{M})^M$, hence
\begin{equation}\label{ch2:eq:probmertongen}
\begin{split}
\P(X_t> \kappa)
&= e^{-\lambda t}  \sum_{n=0}^{M} 
\P\big(N(\mu t+n \alpha, \sigma^2t+n \delta^2)>\kappa \big) \,
\frac{(\lambda t)^n}{n!} \,+\, \cO \Bigg( \bigg( \frac{e\lambda t}{M} \bigg)^M \Bigg) \,,
\end{split}
\end{equation}
where $N(a,b)$ denotes a Gaussian random variable with mean $a$ and variance $b$. 
We recall the standard 
estimate $\log \P(N(0,1) > x) \sim -\frac{1}{2} x^2$ as $x \to \infty$.
Then we can write:
\begin{equation}\label{eq:yiel}
\begin{split}
	& \text{If $t$ is bounded from above (e.g.\ $t \to \bar t \in [0,\infty)$) and $\kappa \to \infty$,}\\
	& \log \P\big(N(\mu t+n \alpha, \sigma^2t+n \delta^2)>\kappa \big) \sim
	-\frac{\kappa^2}{2(\sigma^2t + n\delta^2)} \,.
\end{split}
\end{equation}
In particular, we get from \eqref{ch2:eq:probmertongen} that, for fixed $M \in \N$,
\begin{equation} \label{eq:corre}
\begin{split}
	& \P(X_t > \kappa) \sim e^{-\frac{\kappa^2}{2\sigma^2 t}(1+o(1))} +
	\sum_{n=1}^M e^{-\frac{\kappa^2}{2(\sigma^2 t + n \delta^2)}(1+o(1))} \, \frac{(\lambda t)^n}{n!}
	+ \cO \Bigg( \bigg( \frac{e\lambda t}{M} \bigg)^M \Bigg) \\
	& \le e^{-\frac{\kappa^2}{2\sigma^2 t}(1+o(1))} +
	M \, \max_{n=1,\ldots, M} \, 
	e^{-\big(\frac{\kappa^2}{ 2 n \delta^2}  + n \log \frac{1}{\lambda t}
	+ \log n! \big)(1+o(1))}
	+ \cO \Bigg( \bigg( \frac{e\lambda t}{M} \bigg)^M \Bigg) \,.
\end{split}
\end{equation}
For a lower bound, restricting the sum in \eqref{ch2:eq:probmertongen} to a single value $n\in\N$, 
we get
\begin{equation} \label{eq:correl}
	\P(X_t > \kappa) \ge e^{-\big(\frac{\kappa^2}{2(\sigma^2 t + n \delta^2)}  
	+ n \log \frac{1}{\lambda t}
	+ \log n! \big)(1+o(1))} \,.
\end{equation}

\smallskip

We now prove relation \eqref{eq:antici}.
We fix a family of $(\kappa, t)$ with $t \to 0$
and $\kappa \sim a \, \bkappa_2(t)$ for some $a \in (0,\infty)$.
To get an upper bound, we drop the term $\log n!$ in \eqref{eq:corre}
(since $e^{-\log n!} \le 1$) and
plug $\kappa \sim a \sqrt{\log\frac{1}{t}}$, getting
\begin{equation} \label{eq:middle}
	\P(X_t > \kappa) \le t^{\frac{a^2}{2\sigma^2 t}(1+o(1))}
	+ M \, \max_{n=1,\ldots, M} \, 
	t^{\big(\frac{a^2}{2n \delta^2} + n \big)(1+o(1))}
	+ \cO \Bigg( \bigg( \frac{e\lambda t}{M} \bigg)^M \Bigg) \,.
\end{equation}
Let us denote by $\bar n_a \in \N$ the value of $n\in\N$ for which
the minimum in the definition \eqref{eq:fa} of $f(a)$ is attained.
Choosing $M\in\N$ large enough, so that $M \ge \bar n_a$,
the middle term in \eqref{eq:middle} is $t^{f(a)(1+o(1))}$
and is the dominating one, provided
$M > e \lambda$ and $M > f(a)$, so that
the third term is $\ll t^{f(a)}$.
For an analogous lower bound, we apply \eqref{eq:correl} with $n = \bar n_a$:
since $\sigma^2 t + n \delta^2 \sim n \delta^2$ (recall that $t \to 0$), we get
\begin{equation*}
	\P(X_t > \kappa) \ge  e^{- \log \bar n_a !}
	\, t^{f(a) (1+o(1))} = (const.) \, t^{f(a) (1+o(1))} \,.
\end{equation*}
We have thus proved relation \eqref{eq:antici}.

It remains to prove relation \eqref{eq:antici2}.
We fix a family of $(\kappa, t)$ such that
either $t \to 0$ and $\kappa \gg \bkappa_2(t)$, or $t \to \bar t \in (0,\infty)$
and $\kappa \to \infty$.
Since $n! \ge (n/e)^n$,
\begin{equation*}
	\frac{\kappa^2}{2n \delta^2}  + n \log \frac{1}{\lambda t}
	+ \log n! \,\ge\,
	\frac{\kappa^2}{2n \delta^2}  + n \log \frac{n}{e \lambda t}
	\,\ge\, \inf_{x \ge 0} \left\{
	\frac{\kappa^2}{2 \delta^2 x}  + x \log \frac{x}{e \lambda t}
	\right\} \,.
\end{equation*}
By direct computation, the infimum is attained at
\begin{equation}\label{eq:barx}
	\bar x \sim \frac{\kappa}{\delta \sqrt{2\log \frac{\kappa}{t}}} \,,
\end{equation}
which yields
\begin{equation*}
	\frac{\kappa^2}{2n \delta^2}  + n \log \frac{1}{\lambda t}
	+ \log n! \,\ge\, \frac{\kappa}{\delta}\sqrt{2\log \frac{\kappa}{t}}
	\, \big(1+o(1)\big) \,.
\end{equation*}
We now choose $M = \lfloor 3 \bar x \rfloor$ in \eqref{eq:corre}, so that
\begin{align}
	\nonumber
	\P(X_t > \kappa) & \le
	e^{-\frac{\kappa^2}{2\sigma^2 t}(1+o(1))}
	+ 3 \bar x \,
	e^{- \frac{\kappa}{\delta}\sqrt{2\log \frac{\kappa}{t}} (1+o(1))}
	+ \cO \Bigg( \bigg( \frac{\lambda t}{\bar x} \bigg)^{3 \bar x} \Bigg)  \\
	\label{eq:ubapp}
	& \le
	e^{-\frac{\kappa^2}{2\sigma^2 t}(1+o(1))}
	+ e^{- \frac{\kappa}{\delta}\sqrt{2\log \frac{\kappa}{t}} (1+o(1))}
	+ \cO \Big( e^{-3 \bar x  \log \frac{ \bar x}{\lambda t}} \Big)  \,,
\end{align}
where we have absorbed $3 \bar x$ inside the
$o(1)$ term in the exponential, because $\log (3\bar x) = o(\kappa) =
o(\kappa \sqrt{\log \frac{\kappa}{t}})$
by \eqref{eq:barx} (recall that $\kappa \to \infty$).
The dominant contribution to \eqref{eq:ubapp} is given by the middle
term (note that $3 \bar x \log \frac{\bar x}{\lambda t} \sim \frac{3}{2}
\frac{\kappa}{\delta} \sqrt{2 \log \frac{\kappa}{t}}$, always by \eqref{eq:barx}).
For a corresponding lower bound, we apply \eqref{eq:correl} with
$n = \lfloor \bar x \rfloor$: since $\log n! \sim n\log (n/e)$
and $\sigma^2 t + \lfloor \bar x \rfloor \delta^2 
\sim \lfloor \bar x \rfloor \delta^2$ (because $\bar x \to \infty$), we get
\begin{equation*}
\begin{split}
	\P(X_t > \kappa) 
	& \ge e^{-\big(\frac{\kappa^2}{ 2 \delta^2 \lfloor \bar x \rfloor}  
	+ \lfloor \bar x \rfloor \log \frac{1}{\lambda t}
	+ \log \lfloor \bar x \rfloor! \big)(1+o(1))} =
	e^{-\big(\frac{\kappa^2}{ 2 \delta^2 \lfloor \bar x \rfloor} 
	+ \lfloor \bar x \rfloor \log \frac{\lfloor \bar x \rfloor}{e \lambda t}\big) (1+o(1))} \\
	& = e^{- \frac{\kappa}{\delta}\sqrt{2\log \frac{\kappa}{t}} (1+o(1))} \,.
\end{split}
\end{equation*}
We have thus shown that
\begin{equation*}
	\log \P(X_t > \kappa) \sim - \frac{\kappa}{\delta}\sqrt{2\log \frac{\kappa}{t}} \,,
\end{equation*}
completing the proof of relation \eqref{eq:antici2} and
of Lemma~\ref{th:lemma}.\qed

\subsection{Proof of Proposition~\ref{ch2:th:BS}}
\label{ch2:sec:app:BS}

Let us first prove \eqref{ch2:eq:asd1} and \eqref{ch2:eq:asv}.
Since $\phi(d_2) e^\kappa = \phi(d_1)$, cf.\ \eqref{ch2:eq:phiPhi} and \eqref{ch2:eq:BS2},
recalling \eqref{ch2:eq:Mills} we can rewrite the Black\&Scholes formula \eqref{ch2:eq:BS1} as follows:
\begin{equation} \label{ch2:eq:BS3}
	\CBS(\kappa, v) = \phi(d_1)  \big( U(-d_1) - U(-d_2) \big)
	= \phi(d_1) \big( U(-d_1) - U(-d_1+v) \big)\,.
\end{equation}
If $d_1 \to -\infty$, applying \eqref{ch2:eq:Millsas} we get
\begin{equation*}
	U(-d_1) - U(-d_1+v) =
	- \int_{-d_1}^{-d_1+v} U'(z) \, \dd z \sim \int_{-d_1}^{-d_1+v}
	\frac{1}{z^2} \, \dd z =
	\frac{v}{-d_1(-d_1+v)} \,,
\end{equation*}
and \eqref{ch2:eq:asd1} is proved. 
Next we assume that $v \to 0$. By convexity of $U(\cdot)$
(cf.\ Lemma~\ref{ch2:th:Mills}),
\begin{equation*}
	-U'(-d_1+v) \le \frac{U(-d_1) - U(-d_1+v)}{v} \le -U'(-d_1) \,,
\end{equation*}
hence to prove \eqref{ch2:eq:asv} it suffices to show that $U'(-d_1+v) \sim U'(-d_1)$.
To this purpose, by a subsequence argument, 
we may assume that $d_1 \to \overline{d_1} \in \R \cup \{\pm\infty\}$.
Since $d_1 \le \frac{v}{2}$ for $\kappa \ge 0$, when $v \to 0$ necessarily
$\overline{d_1} \in [-\infty,0]$. If $\overline{d_1} = -\infty$, i.e.\
$-d_1 \to +\infty$, then
$-d_1+v \sim -d_1 \to +\infty$ and $U'(-d_1+v) \sim U'(-d_1)$ follows by \eqref{ch2:eq:Millsas}.
On the other hand, if $\overline{d_1} \in (-\infty,0]$ then
both $U'(-d_1)$ and $U'(-d_1+v)$ converge to $U'(-\overline{d_1}) \ne 0$, by continuity of $U'$,
hence $U'(-d_1)/ U'(-d_1+v) \to 1$, i.e.\ $U'(-d_1+v) \sim U'(-d_1)$ as requested.

\smallskip

Let us now prove \eqref{ch2:eq:cto0}. Assume that $\min\{d_1, \log v\} \to -\infty$,
and note that for every subsequence
we can extract a sub-subsequence along which either $d_1 \to -\infty$ or $v \to 0$.
We can then apply \eqref{ch2:eq:asd1} and \eqref{ch2:eq:asv} to show that $\CBS(\kappa, v)\to 0$:
\begin{itemize}
\item if $d_1 \to -\infty$, the right hand side of \eqref{ch2:eq:asd1} is bounded from above by
$\phi(d_1)/(-d_1) \to 0$;

\item If $\kappa \ge 0$ and $v \to 0$, then
$d_1 \le \frac{v}{2} \to 0$ and consequently
$\phi(d_1) U'(-d_1)$ is uniformly bounded from above, hence
the right hand side of \eqref{ch2:eq:asv} vanishes (since $v\to 0$).
\end{itemize}

Finally, we assume that $\min\{d_1, \log v\} \not\to -\infty$ and show that
$\CBS(\kappa, v) \not\to 0$. Extracting a subsequence, we have 
$\min\{d_1, \log v\} \ge -M$ for some fixed $M \in (0,\infty)$, 
i.e.\ both $v \ge \epsilon :=e^{-M} > 0$ and $d_1 \ge -M$,
and we may assume that $v \to \overline v \in [\epsilon, +\infty]$
and $d_1 \to \overline{d_1} \in [-M,+\infty]$. Consider
first the case $\overline v = +\infty$,
i.e.\ $v \to +\infty$: by \eqref{ch2:eq:BS2} one has $-d_1 + v = -d_2 
\ge \frac{v}{2} \to +\infty$, hence $\phi(d_1) U(-d_1 + v) \to 0$
(because $\phi$ is bounded), and
recalling \eqref{ch2:eq:Mills} relation \eqref{ch2:eq:BS3} yields
\begin{equation*}
	\CBS(\kappa,v) = \Phi(d_1) - \phi(d_1) U(-d_1+v) \to \Phi(\overline{d_1})  > 0 \,.
\end{equation*}
Next consider the
case $\overline v < +\infty$: since $d_1 \le \frac{v}{2}$, we have
$\overline{d_1} \le \frac{\overline{v}}{2}$ and again by \eqref{ch2:eq:BS3} we obtain
$\CBS(\kappa,v) \to \phi(\overline{d_1}) (U(-\overline{d_1}) - U(-\overline{d_1}+\overline{v})) > 0$.
In both cases, $\CBS(\kappa,v) \not\to 0$.\qed

\section*{Acknowledgments}

We thank Fabio Bellini, Stefan Gerhold and Carlo Sgarra for fruitful discussions.


\end{document}